\newif\iflong

\documentclass[USenglish, a4paper,numberwithinsect, cleveref, autoref, nameinlink, numberwithinsect]{lipics-v2021}

\nolinenumbers

\hideLIPIcs

\usepackage{todonotes}

\usepackage[utf8]{inputenc}
\usepackage[arrow, matrix, curve]{xy}
\usepackage{amsthm}
\usepackage{amsmath}
\usepackage{booktabs}
\usepackage{graphicx}
\usepackage{float}
\usepackage{dsfont}
\usepackage{braket}
\usepackage[noadjust]{cite}
\usepackage{tikz}
\usepackage[noend]{algpseudocode}
\usepackage{algorithm}
\usepackage{bm}
\usepackage{framed}
\usepackage{enumitem}
\usetikzlibrary{decorations.pathreplacing, calc,backgrounds}

\nolinenumbers

\hypersetup{colorlinks = true, citecolor=blue}

\newtheorem{observationProof}{Observation}

\Crefname{theorem}{Theorem}{Theorems}
\crefname{theorem}{Thm.}{Thms.}
\Crefname{proposition}{Proposition}{Propositions}
\crefname{proposition}{Prop.}{Props.}
\Crefname{corollary}{Corollary}{Corollaries}
\crefname{corollary}{Cor.}{Cors.}

\newcommand{\PCTT}{\prb{Timeline (Partial) Vertex Cover}}
\newcommand{\PCTlong}{\prb{Timeline Partial Vertex Cover}}
\newcommand{\CTlong}{\prb{Timeline Vertex Cover}}

\newcommand{\PDTT}{\prb{Timeline (Partial) Dominating Set}}
\newcommand{\PDTlong}{\prb{Timeline Partial Dominating Set}}
\newcommand{\DTlong}{\prb{Timeline Dominating Set}}

\newcommand{\PCT}{\prb{Timeline PVC}}
\newcommand{\CT}{\prb{Timeline VC}}

\newcommand{\PDT}{\prb{Timeline PDS}}
\newcommand{\DT}{\prb{Timeline DS}}

\DeclareMathOperator{\W}{W}

\DeclareMathOperator{\DP}{DP}

\DeclareMathOperator{\imw}{imw}
\DeclareMathOperator{\vimw}{vimw}
\usepackage[]{xspace}
\newcommand{\prb}[1]{\textsc{#1}\xspace}

\makeatletter
\newcommand*{\bdiv}{  \nonscript\mskip-\medmuskip\mkern5mu  \mathbin{\operator@font div}\penalty900\mkern5mu  \nonscript\mskip-\medmuskip
}
\makeatother

\newcommand{\DS}{\prb{Dominating Set}}

\newcommand{\MTINF}{\prb{MinTimeline$_\infty$}}

\newcommand{\mt}{\mathcal{T}}
\newcommand{\mg}{\mathcal{G}}

\newcommand{\snapdelta}{\Delta_{\mathrm{max}}}
\newcommand{\Oh}{\mathcal{O}}

\newcommand{\curr}{\mathrm{curr}}
\newcommand{\pos}{\mathrm{pos}}

\newcommand{\prob}[3]{
\begin{framed}
	\noindent \prb{#1} \\
	\textbf{Input:} #2 \\
	\textbf{Question:} #3
\end{framed}
}

\title{Timeline Problems in Temporal Graphs:\\ Vertex Cover vs.~Dominating Set}
\titlerunning{Timeline Problems in Temporal Graphs: Vertex Cover vs.~Dominating Set}
\authorrunning{Anton Herrmann,
Christian Komusiewicz,
Nils Morawietz,
Frank Sommer}

\author{Anton Herrmann}{Algorithmics and Computational Complexity, Technische Universität Berlin, Germany}{a.herrmann@tu-berlin.de}{https://orcid.org/0009-0008-8473-9043}{}

\author{Christian Komusiewicz}{Institute of Computer Science, Friedrich Schiller University Jena, Germany}{c.komusiewicz@uni-jena.de}{https://orcid.org/0000-0003-0829-7032}{}

\author{Nils Morawietz}{Institute of Computer Science, Friedrich Schiller University Jena, Germany\\ LaBRI, Université de Bordeaux, France}{nils.morawietz@uni-jena.de}{https://orcid.org/0000-0002-7283-4982}{Supported by the French ANR, project ANR-22-CE48-0001 (TEMPOGRAL).}

\author{Frank Sommer}{Institute of Logic and Computation, TU Wien, Austria\\ Institute of Computer Science, Friedrich Schiller University Jena, Germany}{frank.sommer@uni-jena.de}{https://orcid.org/0000-0003-4034-525X}{Supported by the Alexander von Humboldt Foundation.}

\Copyright{Anton Herrmann, Christian Komusiewicz, Nils Morawietz, Frank Sommer}



\relatedversion{Some of the results of this work are also contained in the first author’s Master's thesis~\cite{Herrmann24}.}

\begin{document}
\maketitle
\begin{abstract}
A temporal graph is a finite sequence of graphs, called snapshots, over the same vertex set.
Many temporal graph problems turn out to be much more difficult than their static counterparts.
One such problem is~\textsc{Timeline Vertex Cover} (also known as~\textsc{MinTimeline$_\infty$}), a temporal analogue to the classical~\textsc{Vertex Cover} problem.
In this problem, one is given a temporal graph~$\mathcal{G}$ and two integers~$k$ and~$\ell$, and the goal is to cover each edge of each snapshot by selecting for each vertex at most~$k$ activity intervals of length at most~$\ell$ each.
Here, an edge~$uv$ in the~$i$th snapshot is covered, if an activity interval of~$u$ or~$v$ is active at time~$i$.
In this work, we continue the algorithmic study of~\textsc{Timeline Vertex Cover} and introduce the~\textsc{Timeline Dominating Set} problem where we want to dominate all vertices in each snapshot by the selected activity intervals.

We analyze both problems from a classical and parameterized point of view and also consider partial problem versions, where the goal is to cover (dominate) at least~$t$~edges (vertices) of the snapshots.
With respect to the parameterized complexity, we consider the temporal graph parameters vertex-interval-membership-width~$(\vimw)$ and interval-membership-width~$(\imw)$.
We show that all considered problems admit FPT-algorithms when parameterized by~$\vimw + \, k+\ell$.
This provides a smaller parameter combination than the ones used for previously known FPT-algorithms for~\textsc{Timeline Vertex Cover}.
Surprisingly, for~$\imw+\, k+\ell$, \textsc{Timeline Dominating Set} turns out to be easier than~\textsc{Timeline Vertex Cover}, by also admitting an FPT-algorithm, whereas the vertex cover version is NP-hard even if~$\imw+\, k+\ell$ is constant.
We also consider parameterization by combinations of~$n$, the vertex set size, with~$k$ or~$\ell$ and parameterization by~$t$.
Here, we show for example that both partial problems are fixed-parameter tractable for~$t$ which significantly improves and generalizes a previous result for a special case of \textsc{Partial Timeline Vertex Cover} with~$k=1$.
\end{abstract}

\newpage

\section{Introduction}
\label{chapter:introduction}

A crucial task in the management of wireless sensor networks is to monitor the network by selecting few dedicated sensors that can monitor themselves and other sensors which are close enough to have a direct wireless connection~\cite{PCA18,SSZ02}.
In graph-theoretic terms, this is the famous NP-hard \DS{} problem where we say that a vertex dominates itself and all its neighbors and the task is to select few vertices of the graph such that every vertex is dominated by some selected vertex. In some applications of sensor networks, the network may change over time. Then, instead of a static graph~$G$ the input would be a \emph{temporal graph}~$\mg$, consisting of a set of \emph{snapshots}~$G_i$,  each reflecting the connections at timestep~$i$.
Moreover, in such a scenario each sensor could carry out the surveillance only for a bounded duration, for example due to limited battery capacities.
Thus, instead of selecting few sensors to monitor the network, we would like to select, for each sensor, an active time interval of bounded length~$\ell$, during which it monitors itself and all its current neighbors.
To make the model more general, we may also allow for each sensor to select up to~$k$ such active intervals, for example because each sensor carries~$k$ batteries. By calling a collection of active intervals for all vertices of the graph a \emph{$k$-activity $\ell$-timeline}, the scenario described above corresponds to the following problem.
\prob{\DTlong (\DT)}{A temporal graph $\mathcal{G}$ and integers $k \geq 1, \ell \geq 0$.}{Is there a $k$-activity $\ell$-timeline $\mathcal{T}$ which dominates all temporal vertices of~$\mathcal{G}$?}
For an example input and solution for \DT, refer to \Cref{fig:dom-cov}.
As in the static \DS{} problem, we may further generalize the problem to handle scenarios where we are not able to monitor the whole network over the full time period but instead we want to maximize the number of monitored sensors under the resource limitations.
\prob{\PDTlong (\PDT)}{A temporal graph $\mathcal{G}$ and integers $k \geq 1, \ell \geq 0,t \geq 0$.}{Is there a $k$-activity $\ell$-timeline $\mathcal{T}$ which dominates at least $t$ temporal vertices of~$\mathcal{G}$?}

In this work, we initialize the study of \PDTT in particular from a computational complexity perspective.
As we show, \PDTT is NP-hard. 
To cope with this intractability, we consider mostly the parameterized complexity of the problem, with a particular focus on structural parameterizations of temporal graphs.

The idea to consider timelines in a temporal graph setting is not new:  Rozenshtein et al.~\cite{DBLP:journals/datamine/RozenshteinTG21} introduced it in
\CTlong{},\footnote{Rozenshtein et al. denote this problem as \MTINF.} the second main problem studied in this work. 
\prob{\CTlong (\CT)}{A temporal graph $\mathcal{G}$ and integers $k \geq 1, \ell \geq 0$.}{Is there a $k$-activity $\ell$-timeline $\mathcal{T}$ which covers all temporal edges of~$\mathcal{G}$?}For an example input and solution for \CT, refer again to \Cref{fig:dom-cov}. 
The model behind \CT is that edges in a temporal graph arise only when at least one of their endpoints is active and the task of the problem is to provide an explanation of all edges such that the vertices are only active a few times, each time only for a short duration~\cite{DBLP:journals/datamine/RozenshteinTG21}.
\begin{figure}[t]
	\centering
	\begin{tikzpicture}[scale = 0.5, transform shape,
		V/.style = {circle, draw, fill=black}
		]

				\node (G1) at (-13, 0) {\Large $v_1$};
		\node (G1) at (-13, 1.25) {\Large $v_2$};
		\node (G1) at (-13, 2.5) {\Large $v_3$};
		\node (G1) at (-13, 3.75) {\Large $v_4$};
		\node (G1) at (-13, 5) {\Large $v_5$};

				\node[V] (A1) at (-10,0)    {};
		\node[V]  (B1) at (-10, 1.25)       {};
		\node[V]  (C1) at (-10, 2.5)   {};
		\node[V]  (D1) at (-10,3.75)  {};
		\node[V]  (E1) at (-10,5)       {};
		\node (G1) at (-10, -1.5) {\Large $G_1$};
								\draw [very thick] (B1) to [bend left = 45] (D1);
		\draw [very thick] (C1) to (D1);
		\draw [very thick] (D1) to (E1);

				\node[V] (A2) at (-6,0)    {};
		\node[V]  (B2) at (-6, 1.25)       {};
		\node[V]  (C2) at (-6, 2.5)   {};
		\node[V]  (D2) at (-6,3.75)  {};
		\node[V]  (E2) at (-6,5)       {};
		\node (G2) at (-6, -1.5) {\Large $G_2$};
				\draw [very thick] (A2) to (B2);
		\draw [very thick] (B2) to (C2);
		\draw [very thick] (B2) to [bend left = 45] (D2);
		\draw [very thick] (B2) to [bend right = 45] (E2);
				\draw [very thick] (D2) to (E2);

				\node[V] (A3) at (-2,0)    {};
		\node[V]  (B3) at (-2, 1.25)       {};
		\node[V]  (C3) at (-2, 2.5)   {};
		\node[V]  (D3) at (-2,3.75)  {};
		\node[V]  (E3) at (-2,5)       {};
		\node (G3) at (-2, -1.5) {\Large $G_3$};
						\draw [very thick] (A3) to (B3);
		\draw [very thick] (C3) to (D3);
		\draw [very thick] (E3) to (D3);
		\draw [very thick] (C3) to [bend right = 45] (E3);
		
				\node[V] (A4) at (2,0)    {};
		\node[V]  (B4) at (2, 1.25)       {};
		\node[V]  (C4) at (2, 2.5)   {};
		\node[V]  (D4) at (2,3.75)  {};
		\node[V]  (E4) at (2,5)       {};
		\node (G4) at (2, -1.5) {\Large $G_4$};
						\draw [very thick] (B4) to (C4);
						\draw [very thick] (D4) to (E4);
		\draw [very thick] (B4) to (A4);
		\draw [very thick] (D4) to (C4);

				\node[V] (A5) at (6,0)    {};
		\node[V]  (B5) at (6, 1.25)       {};
		\node[V]  (C5) at (6, 2.5)   {};
		\node[V]  (D5) at (6,3.75)  {};
		\node[V]  (E5) at (6,5)       {};
		\node (G5) at (6, -1.5) {\Large $G_5$};
				\draw [very thick] (A5) to (E5);
		\draw [very thick] (C5) to [bend right = 45] (E5);
		\draw [very thick] (C5) to [bend right = 45] (A5);
				
				\node[V] (A6) at (10,0)    {};
		\node[V]  (B6) at (10, 1.25)       {};
		\node[V]  (C6) at (10, 2.5)   {};
		\node[V]  (D6) at (10,3.75)  {};
		\node[V]  (E6) at (10,5)       {};
		\node (G6) at (10, -1.5) {\Large $G_6$};
				\draw [very thick] (A6) to [bend left = 45] (D6);
		\draw [very thick] (B6) to [bend right = 45] (E6);
		\draw [very thick] (C6) to [bend right = 30] (E6);
		\draw [very thick] (B6) to (A6);
		
				\draw[dashed] (-8, -1) -- (-8, 6);
		\draw[dashed] (-4, -1) -- (-4, 6);
		\draw[dashed] (0, -1) -- (0, 6);
		\draw[dashed] (4, -1) -- (4, 6);
		\draw[dashed] (8, -1) -- (8, 6);

				\begin{scope}[on background layer]
			\draw[color=blue, very thick] (-10.5,3.35) rectangle (-1.5, 4.15);
			\draw[color=blue, very thick] (-6.5, 0.85) rectangle (2.5,1.65);
			\draw[color=blue, very thick] (1.5,4.6) rectangle (10.5,5.4);
			\draw[color=blue, very thick] (1.5,-0.4) rectangle (10.5,0.4);
			\draw[color=blue, very thick] (-2.5,2.1) rectangle (6.5,2.9);

			\draw[color=red, very thick] (-10.7,3.2) rectangle (-9.3, 4.3);
			\draw[color=red, very thick] (9.3,3.2) rectangle (10.7, 4.3);

			\draw[color=red, very thick] (-10.7,3.2 - 3.75) rectangle (-9.3, 4.3 - 3.75);
			\draw[color=red, very thick] (-2.7,3.2 - 3.75) rectangle (-1.3, 4.3 - 3.75);

			\draw[color=red, very thick] (-2.7,3.2 - 1.25) rectangle (-1.3, 4.3 - 1.25);
			\draw[color=red, very thick] (5.3,3.2 - 1.25) rectangle (6.7, 4.3 - 1.25);

			\draw[color=red, very thick] (1.3,3.2 + 1.25) rectangle (2.7, 4.3 + 1.25);
			\draw[color=red, very thick] (9.3,3.2 + 1.25) rectangle (10.7, 4.3 + 1.25);

			\draw[color=red, very thick] (1.3,3.2 - 2.5) rectangle (2.7, 4.3 - 2.5);
			\draw[color=red, very thick] (-6.7,3.2 - 2.5) rectangle (-5.3, 4.3 - 2.5);
		\end{scope}
	\end{tikzpicture}
	\caption{The short red boxes represent a 2-activity 0-timeline dominating all temporal vertices and the long blue boxes represent a 1-activity 2-timeline covering all temporal edges. Observe that the interval length is defined in such a way that an interval starting and ending in the same snapshot has length~0.}
	\label{fig:dom-cov}
\end{figure}
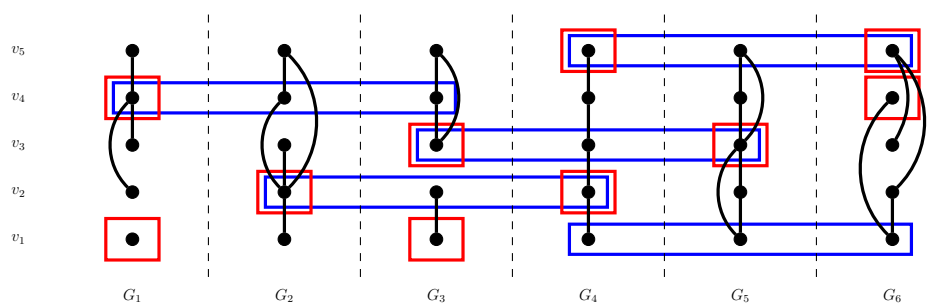
Dondi et al.~\cite{DondiPartialBounded} later analyzed the partial variant of the problem.
\prob{\PCTlong (\PCT)}{A temporal graph $\mathcal{G}$ and integers $k \geq 1, \ell \geq 0,t \geq 0$.}{Is there a $k$-activity $\ell$-timeline $\mathcal{T}$ which covers at least $t$ temporal edges of~$\mathcal{G}$?}
We continue the study of \PCTT with two aims: First, some of the parameters considered in our study have not yet been studied for this problem. Second, we seek a comprehensive complexity overview for these two fundamental timeline problems, highlighting their similarities and differences.

\subparagraph{Previous results.}
Rozenshtein et al.~\cite{DBLP:journals/datamine/RozenshteinTG21} showed that \CT is NP-hard in general, but solvable in polynomial time for $k=1$. Froese et al.~\cite{DBLP:journals/mst/FroeseKZ24} continued the algorithmic study of~\CT.
They proved NP-hardness even if the input consists of at most three snapshots,~${k=2}$ and~$\ell = 0$ and studied the parameterized complexity with respect to different parameters.
Froese et al.~\cite{DBLP:journals/mst/FroeseKZ24} showed fixed-parameter tractability for~$n +k$ and $\W$[1]-hardness for $n + \ell$.
Herein,~$n$ is the number of vertices of the underlying graph which in turn is the static graph obtained by taking the union of all edge sets.
Moreover, \CT admits an XP-algorithm for~$n$ and is fixed-parameter tractable with respect to $n$ if~$\ell = 0$~\cite{DBLP:journals/mst/FroeseKZ24}.
Schubert \cite{SchubertMA} later analyzed the parameterized complexity of  \CT with respect to combinations of input parameters with structural parameters of the underlying graph.
Finally, Dondi et al.~\cite{DondiPartialBounded} obtained two results for~\PCT with~$k=1$: First, it was shown that this case is NP-hard, in contrast to \CT. Second, Dondi et al.~\cite{DondiPartialBounded} gave FPT-algorithms for parameterization by the number of covered edges and by the number of uncovered edges.

\subparagraph{Our results.}
We start off with a new hardness result for \CT{}: We show that the problem remains NP-hard even if the total number of snapshots~$T$,~$k$, and~$\ell$ are constant and additionally the maximum degree in every snapshot is at most one.
Then, we show the NP-hardness of \DT, again for the case that~$T$,~$k$, and~$\ell$ are constant and the maximum snapshot degree is one.

We then study the parameterized complexity of the problems; an overview of the results is given in \Cref{tab:paraRes}. As noted above, \CT{} is fixed-parameter tractable with respect to~$n+k$~\cite{DBLP:journals/mst/FroeseKZ24}. Since~$n$ is a rather large parameter, Froese et al.~\cite{DBLP:journals/mst/FroeseKZ24} asked whether there are FPT-algorithms also for parameters that are smaller than~$n$. We make progress in this direction by considering the two parameters \textit{interval-membership-width $(\imw)$} and \textit{vertex-interval-membership-width $(\vimw)$} which where introduced by Bumpus and Meeks~\cite{vimwParameter}. The parameters $\imw$ and $\vimw$ are interesting in the sense that they are among the relatively few truly temporal structural graph parameters because they are sensitive to the order of the snapshots~\cite{vimwParameter}.
Informally, $\vimw$ can be defined as follows:
We say the \textit{lifetime} of a vertex is the time interval between its first and last non-isolated appearance in a snapshot. For a timestep~$i$, the bag of~$i$ is the set of vertices whose lifetime contains timestep~$i$, and $\vimw$ is the size of the largest bag of~$\mg$.

We first show that \PCT{} is fixed-parameter tractable with respect to $\vimw + k + \ell$. To further improve on this, we introduce a hierarchy of parameters with non-increasing size.
The idea, inspired by the $h$-index of static graphs~\cite{ES12}, is that a temporal graph may have only few large bags.
In that case, we would rather parameterize by the size of the~$x$th largest bag, denoted by~$\vimw[x]$, for some~$x>1$.
 With this definition, we have~$\vimw=\vimw[1]$ and~$\vimw[x]\ge\vimw[x+1]$ for all~$x\in [T]$. We show that \PCT{} is fixed-parameter tractable for~$\vimw[x]+k+\ell$ with~$x=k\cdot(\ell+1)$.
 Informally, this means that the larger~$k$ and~$\ell$ get, the more large bags can be ignored for the parameter.  We then consider \PDTT in the same setting.
 We first showing fixed-parameter tractability of~\PDT for $\vimw+\,k+\ell$. Then we show that for~\DT{} it can be improved to fixed-parameter tractability for $\vimw[x]+k+\ell$ for~$x=\ell+2$ while for \PDT parameterization by~$\vimw[2]+k+\ell$ is not useful. The latter result is obtained by showing NP-hardness of the extremely restricted special case when there are two snapshots, one of which is edgeless,~$k=1$ and~$\ell=0$.

 Afterwards, we consider the interval-membership-width~$(\imw)$. Informally, this is the edge-version of $\vimw$, that is, the bag at timestep~$i$ contains the edges which occur at timestep~$i$ and those that occur both before and after~$i$. Note that~$\imw$ can be considered to be  a smaller parameter than~$\vimw$: For all instances, we have~$\vimw \in \Omega(\sqrt{\imw})$ and there are instances with constant values of~$\imw$ and unbounded values of~$\vimw$.
 We show that \DT is fixed-parameter tractable for the parameter~$\imw+\,k+\ell$ while all three other problems are NP-hard for constant values of $\imw+\,k+\ell$. Given the latter hardness results, we refrain from analyzing a hierarchy of~$\imw$-based parameters. Altogether, the results show that in our setting, $\vimw$ is a much more powerful parameter than~$\imw$.

We then conclude by considering the more standard parameters~$n$, $k$, $\ell$, and~$t$, where~$t$ denotes the number of vertices to dominate or edges to cover, respectively.
For~\DT, we show fixed-parameter tractability for~${n+k}$ and $\W$[1]-hardness for $n+ \ell$, thus showing that the complexity is similar to the one of \CT.
Finally, we show that \PDT and \PCT can be solved in $2^{\Oh(t)} \cdot (n+T)^{\Oh(1)}$~time.
This improves and extends a previous result of
 Dondi et al.~\cite[Thm.~6]{DondiPartialBounded} who provided an algorithm for~\PCT with~$k=1$ that has running time~$2^{t\cdot \log(t)} \cdot (n+T)^{\Oh(1)}= t^t \cdot (n+T)^{\Oh(1)}$.

\begin{table}[t]
	\caption{Overview of our results on the classic and parameterized complexity for  \PCTT and \PDTT.
	Herein, $\vimw$ denotes the vertex-interval-membership-width, $\imw$ denotes the interval-membership-width, and $\vimw[x]$ denotes the size of the $x$-largest bag in the vertex-interval-membership sequence.
	Recall that~$\vimw=\vimw[1]$.}
	\label{tab:paraRes}
	\centering
	\small
	\begin{tabular}{@{}l r r r r}
	\toprule

		 Parameter & \CT & \PCT & \DT & \PDT \\

\midrule

\multirow{2}{*}{$\vimw + \,k +\ell$} & FPT & FPT & FPT & FPT \\
		& \scriptsize \cref{CovTimeLine_vimw+k+l} & \scriptsize \cref{CovTimeLine_vimw+k+l} & \scriptsize \cref{DomTimeLine_vimw+k+l} & \scriptsize \cref{DomTimeLine_vimw+k+l}\\

		\midrule

		\multirow{2}{*}{$\imw +\, k +\ell$} & NP-h & NP-h  & FPT & NP-h  \\
		& \scriptsize \cref{thm-cov-np-h-imw-k-l} & \scriptsize \cref{thm-cov-np-h-imw-k-l} & \scriptsize \cref{thm-kernel-imw-k-l} & \scriptsize \cref{thm-pdt-np-h-imw-k-l}
		\\

\midrule

		$k+\ell+\vimw[x]$ & $x = k(\ell+1)+1$ & $x = k(\ell+1)+1$ & $x = \ell+2$ & $x = 2$ \\
		& FPT & FPT & FPT & NP-h \\
& \scriptsize \cref{CovTimeLine_vimw+k+l2} & \scriptsize \cref{CovTimeLine_vimw+k+l2} & \scriptsize \cref{DomTimeLine_vimw+k+l2} & \scriptsize \cref{PDomTimeLine_vimw2+k+l} \\

		\midrule
		\midrule

\multirow{2}{*}{$n+k$} & FPT & \multirow{2}{*}{?} & FPT & \multirow{2}{*}{?} \\
		& \scriptsize\cite[Thm.~8]{DBLP:journals/mst/FroeseKZ24} & & \scriptsize \cref{DomTimeline_n+k} & \\

		\midrule

		\multirow{2}{*}{$n+\ell$} & W[1]-h & W[1]-h & W[1]-h & W[1]-h \\
		& \scriptsize\cite[Thm.~12]{DBLP:journals/mst/FroeseKZ24} & \scriptsize\cite[Thm.~12]{DBLP:journals/mst/FroeseKZ24} & \scriptsize \cref{W1hardnessN} & \scriptsize\cref{W1hardnessN}\\

		\midrule

		\multirow{2}{*}{$t$} & \multirow{2}{*}{-} & FPT & \multirow{2}{*}{-} & FPT \\
		& & \scriptsize
		\cref{CovTimelineParametert} 		 & & \scriptsize
		 \cref{DomTimelineParametert}\\

		\bottomrule
	\end{tabular}
\end{table}

\subparagraph{Further related work.}

Akrida et al.~\cite{DBLP:journals/jcss/AkridaMSZ20} introduced the problems \prb{Temporal Vertex Cover (TVC)} and \prb{Sliding Window Temporal Vertex Cover (SW-TVC)}, where the goal is to find a minimum number of temporal vertices that cover the edges of the temporal graph.
In \prb{TVC} each edge needs to be covered in at least one of the snapshots and in \prb{SW-TVC}, the input contains an integer $\Delta$, and the aim is to cover every edge at least once at every $\Delta$ consecutive time steps.
Akrida et al.~\cite{DBLP:journals/jcss/AkridaMSZ20} showed that \prb{TVC} remains NP-hard when the underlying graph is a star graph.
Moreover, they studied the approximability of both problems, provided (S)ETH-based running time lower bounds and designed an exact dynamic programming algorithm with exponential running time.
The research on these problems was then continued by Hamm et al.~\cite{DBLP:conf/aaai/HammKMS22}, who showed that \prb{TVC} is solvable in polynomial time if the underlying graph is a path or cycle while \prb{SW-TVC} remains NP-hard in this case.
Besides that, they provided several (exact and approximation) algorithms for the sliding window variant. The \prb{Temporal Dominating Set (TDS)} problem is defined analogously, here the task is to dominate every vertex in at least one of the snapshots of the temporal graph by selecting few temporal vertices. \prb{TDS} is NP-hard in very restricted cases, for example when  the maximum degree in every snapshot is 2~\cite{HKMS25a}.
There are further temporal variants of \prb{Vertex Cover} and \prb{Dominating Set} which have been discussed and studied. These include a multistage variant of \prb{Vertex Cover} \cite{MultistageFluschnik}, a reachability-based variant for \prb{Dominating Set} \cite{TemporalDS2} and different conditions on how and when a vertex should be dominated \cite{TemporalDS1, TemporalDS3}.

\section{Preliminaries}
\label{sec:preliminaries}

We denote the set of integers $\{i, i+1, \dots, j-1, j \}$ by $[i,j]$, and $[1,i]$ simply by $[i]$.
Given a set $V$ and an integer $k > 0$, we write $\binom{V}{k}$ for the collection of all size-$k$ subsets of~$V$.

\subparagraph{Graph theory.}
A \emph{(static) graph}~$G = (V, E)$ consists of a set of \emph{vertices}~$V$ and a set of edges~$E \subseteq \binom{V}{2}$.
We denote by~$V(G)$ and~$E(G)$ the vertex and edge set of~$G$, respectively. Furthermore, we let~$n := |V(G)|$ and~$m := |E(G)|$.
For an edge~$\{u,v\}$ we write~$uv$ and call~$u$ and~$v$ \emph{endpoints} of the edge.
Further, we say that the vertex~$v \in V$ is \emph{incident} with~$e \in E$ if~$v \in e$.
If~$uv \in E$, then~$u$ and~$v$ are \emph{adjacent} and they are \emph{neighbors} of each other.
The set~$N_G(v):=\{u\in V: u \text{ and } v \text{ are adjacent}\}$ is the \emph{(open) neighborhood} of~$v$.
Moreover, $N_G[v]:=N_G(v)\cup\{v\}$ is the \emph{closed neighborhood} of~$v$.
For a set of vertices~$S$, we let~$N_G[S]:= \bigcup_{s\in S} N_G[s]$ and~$N_G(S)=N_G[S] \setminus S$.
By~$\deg_G(v):= |N_G(v)|$ we denote the \emph{degree} of~$v$.
If~$\deg_G(v)=0$, we say that~$v$ is \emph{isolated}.
The \emph{maximum degree} is~$\Delta (G) := \max_{v \in V} \deg_G (v)$.
For vertex sets~$S$ and~$T$, we let~$E_G(S,T):=\{uv\in E(G): u\in S \text{ and } v\in T\}$ and we let~$E(S):=E(S,S)$.
A graph~$G$ is \emph{bipartite} if~$V(G)$ can be partitioned into two sets~$S$~and~$T$ such that~$E(G)=E(S,T)$.

A \emph{star} is a bipartite graph where~$|S|=1$.
We may drop the subscript~$\cdot_G$ when it is clear from context.
A \emph{path} in~$G$ is a sequence of distinct vertices~$(v_1, v_2, \ldots, v_{z+1})$ such that~$v_iv_{i+1}\in E$ for each~$i\in[z]$.
In this case we say that each vertex~$v_i$ is \emph{connected} to each vertex~$v_j$.

For more details on graph theory, we refer to the book by Diestel~\cite{DBLP:books/daglib/0030488}.

\subparagraph{Temporal graphs.}
	A \emph{temporal graph}~$\mathcal{G}$ is a finite sequence of graphs~$(G_1, \dots, G_T)$ which all have the same set of vertices~$V(\mathcal{G}) := V(G_1) = \dots = V(G_T)$.
	The graphs~$G_i$ are called \emph{snapshots}, $T$ is the \emph{lifetime of~$\mg$}, and~$i\in[T]$ is a \emph{time step}.
	We may write~$V$ instead of~$V(\mathcal{G})$ and~$E_i$ instead of~$E(G_i)$.
Moreover, for~$u,v\in V$ and~$i\in[T]$, $(v,i)$ is a \emph{temporal vertex} and~$(\{u,v\},i) = (uv,i)$ with~$uv \in E_i$ is a \emph{temporal edge} in snapshot~$G_i$. We say edge~$e$ \emph{appears} in~$G_i$ if~$e \in E_i$.
Moreover, we say that a vertex~$v$ is \emph{incident} with the temporal edge~$(e,i)$ if~$v \in e$.
For a temporal vertex $(v,i)$, we define the \emph{(open) neighborhood} by~$N_\mathcal{G} ( (v,i) ) := N_{G_i} (v)$ and the \emph{closed neighborhood} by~$N_\mathcal{G} [ (v,i) ] := N_{G_i} [v]$.
By~$G_ \downarrow (\mathcal{G}) := (V_\downarrow := V(\mathcal{G}), E_\downarrow : = \bigcup_{i = 1}^T E_i)$ we denote the \emph{underlying graph} of $\mathcal{G}$.
We let~$\Delta_i (\mathcal{G})=\Delta(G_i)$.
Furthermore, by~$\snapdelta (\mathcal{G}) := \max_{i=1}^T \Delta_i (\mathcal{G})$ we denote the \emph{maximum snapshot degree} of $\mathcal{G}$.
If~$E_i=\emptyset$ we say that~$G_i$ is \emph{empty}.
We may drop the subscript~$\cdot_\mg$ when it is clear from context.

\subparagraph{Parameter definitions.}

We give the precise definitions of the parameters related to the (vertex)-interval-membership width:~$\vimw$,~$\vimw[x]$, and~$\imw$.

Let $\mathcal{G} = (G_1, \dots, G_T)$ be a temporal graph.
The \emph{vertex-interval-membership sequence} of $\mathcal{G}$ is the sequence $(F_i)_{i \in [T]}$ of vertex-subsets $F_i \subseteq V(\mathcal{G})$ (called \emph{bags}), where each $F_i$ is defined as $$F_i := \{v \in V(\mathcal{G}) : \exists p,q \in [T] \text{ such that } p \leq i \leq q,  \deg_{G_p} (v) \geq 1 \text{ and } \deg_{G_q}(v) \geq 1 \}.$$
The \emph{vertex-interval-membership-width} of $\mathcal{G}$, denoted by $\vimw (\mathcal{G})$, is the maximum size of a bag in the vertex-interval-membership sequence, that is, $\vimw (\mathcal{G}) := \max \{|F_i| \colon i \in [T]\}$.

For given~$x\in [T]$ we introduce the parameter~$\vimw [x]$ as the size of the~$x$th largest bag of the vertex-interval-membership sequence.
Formally, let~$(F_{i_1}, \dots, F_{i_T})$ be an ordering of~$(F_i)_{i \in [T]}$ such that~$|F_{i_1}| \geq \dots \geq |F_{i_T}|$.
Then we define $$\vimw[x](\mg) := |F_{i_x}|.$$

Note that~$\vimw[i]$ can be much smaller than~$\vimw[1] = \vimw$: for example consider a temporal graph where every vertex only appears non-isolated for a short period of time.
If all snapshots have only a few edges, then~$\vimw$ is small.
However, a single snapshot with many edges is enough to make $\vimw$ arbitrarily large.
If only a few such snapshots exist, then~$\vimw[i]$ is already much smaller then~$\vimw$ for small~$i$.

The \emph{interval-membership sequence} of~$\mg$ is the sequence~$(F_i)_{i \in [T]}$ of edge-subsets~$F_i \subseteq E_\downarrow$ (again called \emph{bags}), where each~$F_i$ is defined as $$F_i := \{e \in E_\downarrow : \exists p,q \in [T] \text{ such that } p \leq i \leq q \text{ and } e \in E_p \cap E_q \}.$$
The \emph{interval-membership-width} of $\mathcal{G}$, denoted by $\imw (\mathcal{G})$, is the maximum size of a bag in the interval-membership sequence, that is, $\imw (\mathcal{G}) := \max \{|F_i|\colon i \in [T]\}$.

We remark that~$\vimw = \Omega(\sqrt{\imw})$ and that there exist temporal graphs with~$\imw = 1$ but arbitrary large~$\vimw$~\cite{vimwParameter}.
Finally, note that the vertex-interval-membership sequence and the interval-membership sequence can be computed in polynomial time with respect to the input size.

\subparagraph{Timelines.} Given a temporal graph $\mathcal{G} = (G_1, \dots , G_T)$, an \textit{activity interval} of a vertex $v$ is a triple $(v,a,b)\in V(\mathcal{G}) \times [T] \times [T]$ such that $a \leq b$. We say $a$ that is the \textit{starting time} and  $b$ is the \textit{end time} of an activity interval $(v,a,b)$.
The \textit{length} of the activity interval is defined by $b-a$.
A \textit{$k$-activity timeline} of the temporal graph $\mathcal{G}$ is a set of activity intervals $\mathcal{T} \subseteq \{(v,a,b) \in V(\mathcal{G}) \times [T] \times [T] : a \leq b\}$, which contains at most~$k$ activity intervals for each vertex.
A \textit{$k$-activity $\ell$-timeline} is a $k$-activity timeline in which each activity interval has length at most $\ell$.
Given a $k$-activity timeline~$\mathcal{T}$, a vertex $v$ is called \textit{active} in snapshot $G_i$, if there exists $(v,a,b) \in \mathcal{T}$ such that~$a \leq i \leq b$.
 We say that the activity interval~$(v,a,b)$ \textit{intersects} snapshot $G_i$, if $a \leq i \leq b$. 
A $k$-activity timeline $\mathcal{T}$ \textit{dominates} the temporal vertex $(v,i)$ if~${v \in N_{G_i} [\{ u : u \text{ is active in } G_i \}]}$.
A $k$-activity timeline $\mathcal{T}$ \textit{covers} the temporal edge $(e,i)$ if at least one of the two endpoints of~$e$ is active in snapshot~$G_i$.

\subparagraph{Parameterized complexity.}
Let~$L \subseteq \Sigma^*$ be a computational problem specified over some alphabet~$\Sigma$ and let~$p : \Sigma^* \to \mathds{N}$ be a \emph{parameter}, that is, $p$ assigns to each instance of~$L$ an integer parameter value (which we simply denote by~$p$ if the instance is clear from the context).
We say that~$L$ is \emph{fixed-parameter tractable (FPT) with respect to~$p$}
if~$L$ can be decided in $f(p) \cdot |I|^{\Oh(1)}$~time where~$|I|$ is the length of the input.
The corresponding hardness concept related to fixed-parameter tractability is
W[$t$]-hardness, $t \ge 1$; if a problem~$L$ is W[$t$]-hard with respect to~$p$, then~$L$ is assumed to \emph{not} be fixed-parameter tractable.
Moreover, $L$ is \emph{slice-wise polynomial (XP) with respect to~$p$} if~$L$ can be decided in $f (p) \cdot |I|^{g(p)}$~time.
Let~$(I, p)$ and~$(I', p')$ be two instances of~$L$.
A \emph{reduction to a (problem) kernel for~$L$} is a polynomial-time algorithm that computes an instance~$(I', p')$ such that
\begin{itemize}
\item $p' + |I'| \le g(p)$, and
\item $(I, p)$ is a yes-instance of~$L$ if and only if~$(I', p')$ is a yes-instance of~$L$.
\end{itemize}
The instance~$(I', p')$ is referred to as \emph{the problem kernel} and~$g$ is the \emph{size} of the kernel.
Furthermore, if~$g$ is a polynomial, we say that the kernel is a \emph{polynomial problem kernel}.
For more details about parameterized complexity we refer to the standard monographs~\cite{downey2012parameterized,CFK+15}.

\section{NP-Hardness Results}
\label{HardnessResults}

We show NP-hardness for \CT and \DT even when they are restricted to temporal graphs with a maximum snapshot degree of~1 and constant lifetime~$T$, interval number~$k$, and interval length~$\ell$.
First, we present the reduction for \CT.

\begin{theorem}
	\label{CovTimelineMaxDeg1Hardness}
	\CT is NP-hard even if $T=23, k= 2, \ell =4$, and the maximum snapshot degree is one.
\end{theorem}
\begin{proof}
	We reduce from the NP-hard problem \prb{$3$-Coloring} on graphs with maximum degree four~\cite{DBLP:conf/stoc/GareyJS74}.
In \prb{$3$-Coloring} the input is a graph and the question is whether the vertices can be colored with three colors such that no two adjacent vertices have the same color.

	\textit{Intuition:} The basic idea is to construct a temporal graph with three parts $C_1, C_2$ and~$C_3$ (call them \textit{color blocks}), each consisting of five snapshots and representing one of the three colors.
The part in which a vertex $v$ is \emph{not active} corresponds to the assigned color of $v$.
As each of the three parts will contain every edge of $G$ and only vertices of degree at most one, this ensures that all neighbors are active in the part where $v$ is not active.
Hence, they are assigned a different color.
The basic idea for encoding a given graph $G$ into snapshots of maximum degree one is to find a proper edge coloring with five colors and split the edge set with respect to this coloring, which can be done in polynomial time via Vizing's theorem.
Our aim is that no activity interval of any vertex can hit more than one color block.
To ensure this, we add sufficiently many empty snapshots between any two parts.

	\textit{Construction}: Let $(G= (V = \{v_1, \dots, v_n\},E))$ be an instance of \prb{$3$-Coloring} on graphs with maximum degree four.
We construct the following temporal graph~$\mathcal{G}$ where~$V(\mathcal{G})=V$.

	The temporal graph $\mathcal{G}$ consists of $23$~snapshots.
It can be seen as a sequence of five parts, namely $C_1, H_1, C_2, H_2, C_3$.
Parts~$H_1$ and~$H_2$ consist of $4$~empty snapshots and each of~$C_1,C_2,C_3$ consists of 5~snapshots.
Thus, snapshots~$6$ to~$9$ and~$15$ to~18 are empty.
We denote~$H_1$ and~$H_2$ as the \emph{empty blocks} and we call $C_1,C_2$ and $C_3$ \textit{color blocks} as they correspond to the three colors.
In the following we formally define the snapshots of the color blocks.
All three color blocks~$C_1, C_2$ and~$C_3$ will be identical.
Hence, we only formally define~$C_1$, that is, snapshots~1 to~5.

Here we exploit the fact that the graph~$G$ has maximum degree four, because this implies that there exists a proper edge coloring with at most five colors due to Vizing's theorem~\cite{V64}, which can be computed in polynomial time~\cite{DBLP:journals/ipl/MisraG92}.
Let~$ E(G) = F_1 \cup \dots \cup F_5$ be the partition of~$E(G)$ into five colors of some proper edge coloring.
Now, snapshot~$j$ contains exactly the edges of~$F_j$.
	Finally, observe that by construction, the maximum degree in each snapshot is at most one.
By setting $k =2$ and $\ell = 4$, we obtain the \CT instance $(\mathcal{G}, k, \ell)$.

	\textit{Correctness:}
	We show that $G$ is $3$-colorable if and only if $\mathcal{G}$ has a $2$-activity $4$-timeline~$\mathcal{T}$ covering all temporal edges of~$\mathcal{G}$.
In the following we say a vertex~$v$ is \emph{active} during~$C_j$ for~$j\in[3]$ to indicate that an activity interval of~$v$ starts at the first time step and ends and the last time step of the corresponding part.

$(\Rightarrow)$
Let~$\phi: V(G)\to [3]$ be a proper 3-coloring of~$G$.
We construct an activity timeline of~$\mathcal{G}$ which covers all temporal edges.
For each~$i \in [n]$, the vertex $v\in V(G)$ is active in the two color blocks, which do not correspond to the assigned color of $v$.
Clearly, this yields a 2-activity timeline~$\mathcal{T}$.
It remains to argue that each temporal edge is covered by~$\mathcal{T}$.
Since~$\phi$ is a proper 3-coloring, both endpoints of each edge~$uw\in E(G)$ have a different color.
Moreover, since in color block~$C_j$ all vertices are active which do not have color~$j$, we conclude that all temporal edges of~$\mathcal{G}$ are covered by~$\mathcal{T}$.

$(\Leftarrow)$
Let~$\mathcal{T}$ be a $2$-activity 4-timeline covering all temporal edges of~$\mathcal{G}$.
Observe that since parts~$H_1$ and~$H_2$ corresponding to snapshots~6 to~9 and~15 to~18 are empty, no activity interval of any vertex can hit more than one color block.
Consequently, since~$k=2$, for each vertex~$v$ there exists at least one color~$j\in[3]$ such that~$v$ is not active during any snapshot of color block~$C_j$.
Let~$\phi(v)\in[3]$ be such a color.
Note that if~$\phi(v)$ is not unique, we pick any color fulfilling the property.
We now argue that~$\phi$ is a proper 3-coloring of~$G$.
Assume towards a contradiction that this is not the case, that is, there exists at least one edge~$uw\in E(G)$ such that both endpoints~$u$ and~$w$ have the same color~$j$.
This implies that both~$u$ and~$w$ are not active during~$C_j$.
But since exactly one of the 5~snapshots corresponding to color block~$C_j$ contains the edge~$uw$, timeline~$\mathcal{T}$ is not covering all temporal edges of~$\mathcal{G}$, a contradiction.
Thus, $\phi$ is a proper 3-coloring of~$G$.
\end{proof}
Next, we extend the ideas of this reduction to \DT for which the reduction is more involved since we also need to deal with isolated temporal vertices.

\begin{theorem}

	\label{DomTimelineMaxDeg1Hardness}
	\DT is NP-hard even if $T=35, k= 3, \ell =6$ and the maximum snapshot degree is one.
\end{theorem}

\begin{proof}
	We reduce from \prb{$3$-Coloring} on graphs with maximum degree four.
This problem is known to be NP-hard \cite{DBLP:conf/stoc/GareyJS74}.
In \prb{$3$-Coloring} the input is a graph and the question is whether the vertices can be colored with three colors such that no two adjacent vertices have the same color.
\Cref{Figure3ColorabilityToDomTimeline} sketches the following reduction.
We start by giving an intuition.

	\textit{Intuition:} The basic idea is to construct a temporal graph with three parts $C_1, C_2$ and $C_3$ (call them \textit{color blocks}), each consisting of seven snapshots and representing one of the three colors.
The part in which a vertex $v$ is not active corresponds to the color assigned to~$v$.
As each of the three parts will contain every edge of $G$ and only vertices of degree at most one, this ensures that all neighbors are active in the part where $v$ is not active.
Hence, they are assigned a different color.
The basic idea for encoding a given graph $G$ into snapshots of maximum degree one is to find a proper edge coloring and split the edge set with respect to this coloring.
In the following reduction we unfortunately need in addition to $C_1,C_2$ and $C_3$ another part~$H$ consisting of $14$ snapshots and some additional vertices.
This leads to $k = 3$ for the following reasons.
On the one hand we need to ensure that some activity intervals of the vertices coincide exactly with the color blocks.
On the other hand we need to take care of possible isolated vertices in the snapshots.

	\textit{Construction}: Suppose $(G= (V = \{v_1, \dots, v_n\},E))$ is a given instance of \prb{$3$-Coloring} on graphs with maximum degree four.
We construct the following temporal graph~$\mathcal{G}$.
	For each vertex $v_i \in V(G)$ we introduce five vertices.
Formally, we set
	\begin{align*}
		V(\mathcal{G}) := \bigcup_{i = 1}^n  \, \{v_i, v_i', u_i, u_i', u_i''\}.
	\end{align*}
	The vertex $v_i$ represents the vertex $v_i$ from $G$ in some snapshots of the temporal graph.
The vertex $v'_i$ ensures that two activity intervals of $v_i$ coincide exactly with two of the three color blocks and that $v_i$ is not isolated in any of the snapshots of the color blocks.
The vertices $u_i, u'_i$ and $u''_i$ ensure that $v'_i$ is not isolated in any of the color blocks.
However, the introduction of new vertices also leads to an increasing number of activity intervals.
The later introduced part $H$ of our temporal graph forces us to use them.

	The temporal graph $\mathcal{G}$ consists of $35$ snapshots.
It can be seen as a sequence of four parts, namely $H, C_1, C_2$ and $C_3$.
Part $H$ consists of $14$ snapshots and each of~$C_1,C_2,C_3$ consists of seven snapshots.
We call $C_1,C_2$ and $C_3$ \textit{color blocks} as they correspond to the three colors.
In the following we formally define the snapshots.

	The first part of the temporal graph is the part $H$ in which all snapshots have the same set of edges.
We set $$E_1 = \dots = E_{14} := \bigcup_{i=1}^n \, \{v_iu_i, u_i'u_i''\}$$ while $v_i'$ remains isolated in all these snapshots for each $i \in [n]$.

	The three color blocks $C_1, C_2$ and~$C_3$ are very similar to each other.
We formally define $C_1$ and state how $C_2$ and $C_3$ differ from $C_1$.
	The first snapshot (say $L_1$) and last snapshot (say $R_1$) of $C_1$ are equal.
They ensure that some activity intervals coincide exactly with the color block.
The five snapshots between $L_1$ and $R_1$ encode our graph $G$ with snapshots of maximum degree one.
Here we need the fact that our graph $G$ has maximum degree four, because this implies that there exists a proper edge coloring with at most five colors due to Vizing's theorem~\cite{V64}, which can be computed in polynomial time~\cite{DBLP:journals/ipl/MisraG92}.
Let $ E(G) = F_1 \cup \dots \cup F_5$ be the partition of $E(G)$ into five colors of some proper edge coloring.
With this we are now ready to formally define the snapshots of $C_1$.
Recall that the first $14$ snapshots are part of $H$, so the first and last snapshot of $C_1$ are $G_{15}$ and $G_{21}$ respectively.
We define
	\begin{align*}
		E_{15} = E_{21} := \bigcup_{i=1}^n \,  \{v_iv_i', u_i'u_i''\}.
	\end{align*}
	The five snapshots between them encode $G$, ensure that $v_i$ and $v'_i$ are not isolated and contain the edge $u'_iu_i''$.
Formally, we define for each $j \in [5]$ the edge set of~$G_{15+j}$ by
	\begin{align*}
		& E_{15 + j} := F_j \cup \bigcup_{i = 1}^n \, \{u_i'u_i''\} \cup E_{i,j}
	\end{align*}
	where we set $E_{i,j} = \{v_iv_i'\}$ if $v_i$ is an isolated vertex in the graph $(V(G), F_j)$, and~$E_{i,j} = \{v_i'u_i\}$, otherwise.

	As mentioned above, $C_2$ and $C_3$ are very similar to $C_1$.
Both of them consist of seven snapshots, where the first and last snapshot are equal and the five snapshots between them encode $G$.
The color part $C_2$ starts at snapshot $G_{22} $ and ends at snapshot $G_{28}$.
The color part $C_3$ starts at snapshot $G_{29} $ and ends at snapshot $G_{35}$.
The only difference between $C_1$ and the other color blocks is that $u_i$ swaps roles with $u_i'$ in $C_2$ and with $u_i''$ in~$C_3$ respectively.
The first snapshots of $C_2$ and $C_3$ are called $L_2$ and $L_3$ respectively and the last snapshots of $C_2$ and $C_3$ are called $R_2$ and $R_3$ respectively
	Finally, observe that by construction, the maximum degree in each snapshot is at most one.
By setting $k =3$ and $\ell = 6$, we obtain the \DT instance $(\mathcal{G}, k, \ell)$.
	\begin{figure}[t]
		\centering
		\begin{tikzpicture}[scale = 0.6, transform shape,
			V/.style = {circle, draw, fill=black},
			d/.style = {circle, draw, fill=black, inner sep = 0.5pt},
			d2/.style = {circle, draw, fill=black, inner sep = 1.5pt}
			]

						\node[V, label=above:\large $v_i$] (vi1) at (-11, 4)    {};
			\node[V, label=left:\large $v_i'$] (vi'1) at (-12, 3)    {};
			\node[V, label=right:\large $u_i$]  (ui1) at (-10, 3)       {};
			\node[V, label=left:\large $u_i'$] (ui'1) at (-12, 2)    {};
			\node[V, label=right:\large $u_i''$]  (ui''1) at (-10, 2)       {};
			\node[d] at (-11,1.5)   {};
			\node[d] at (-11,1.2)  {};
			\node[d] at (-11,0.9)  {};
			\node (G1-14) at (-11, -1.5) {\Large $G_1 = \dots = G_{14}$};
						\draw (vi1) edge (ui1);
			\draw (ui'1) edge (ui''1);

						\node[V, label=above:\large $v_i$] (vi2) at (-6, 4)    {};
			\node[V, label=left:\large $v_i'$] (vi'2) at (-7, 3)    {};
			\node[V, label=right:\large $u_i$]  (ui2) at (-5, 3)       {};
			\node[V, label=left:\large $u_i'$] (ui'2) at (-7, 2)    {};
			\node[V, label=right:\large $u_i''$]  (ui''2) at (-5, 2)       {};
			\node[d] at (-6,1.5)   {};
			\node[d] at (-6,1.2)  {};
			\node[d] at (-6,0.9)  {};
			\node (L1) at (-6, -1.5) {\Large $G_{15} = L_1$};
						\draw (vi2) edge (vi'2);
			\draw (ui'2) edge (ui''2);

						\draw (-1,3) ellipse (1.4 and 1.65);
			\node at (-1, 3.8) {\Large $(V, F_j)$};
			\node[V, label=above:\large $v_i$] (vi3) at (-1, 2.05)    {};
			\node[V, label=left:\large $v_i'$] (vi'3) at (-2, 1)    {};
			\node[V, label=right:\large $u_i$]  (ui3) at (0, 1)       {};
			\node[V, label=left:\large $u_i'$] (ui'3) at (-2, 0)    {};
			\node[V, label=right:\large $u_i''$]  (ui''3) at (0, 0)       {};
			\node (M1) at (-1, -1.5) {\Large $G_{15 + j}$ for $j \in [5]$};
						\draw[color=blue, line width=1.5pt] (vi3) edge (vi'3);
			\draw[color=red, line width=1.5pt] (vi'3) edge (ui3);
			\draw (ui'3) edge (ui''3);

						\node[V, label=above:\large $v_i$] (vi4) at (4, 4)    {};
			\node[V, label=left:\large $v_i'$] (vi'4) at (3, 3)    {};
			\node[V, label=right:\large $u_i$]  (ui4) at (5, 3)       {};
			\node[V, label=left:\large $u_i'$] (ui'4) at (3, 2)    {};
			\node[V, label=right:\large $u_i''$]  (ui''4) at (5, 2)       {};
			\node[d] at (4,1.5)   {};
			\node[d] at (4,1.2)  {};
			\node[d] at (4,0.9)  {};
			\node (R1) at (4, -1.5) {\Large $G_{21} = R_1$};
						\draw (vi4) edge (vi'4);
			\draw (ui'4) edge (ui''4);

						\node[d] at (8.8,2)   {};
			\node[d] at (9.1,2)  {};
			\node[d] at (9.4,2)  {};

						\node at (9, -1.5) {\Large $\dots, G_{28}, \dots ,G_{35}$};

						\draw[dashed] (-8.5, -1) -- (-8.5, 5);
			\draw[dashed] (-3.5, -1) -- (-3.5, 5);
			\draw[dashed] (1.5, -1) -- (1.5, 5);
			\draw[dashed] (6.5, -1) -- (6.5, 5);

						\draw [decorate,decoration={brace,amplitude=10pt,mirror,raise=4pt}, line width=0.5pt] (-13,-2) -- (-9,-2);
			\node at (-11, -3.5) {\LARGE $H$};
			\draw [decorate,decoration={brace,amplitude=10pt,mirror,raise=4pt}, line width=0.5pt] (-8,-2) -- (6,-2);
			\node at (-1, -3.5) {\LARGE $C_1$};
			\draw [decorate,decoration={brace,amplitude=10pt,mirror,raise=4pt}, line width=0.5pt] (7,-2) -- (8.9,-2);
			\node at (7.95, -3.5) {\LARGE $C_2$};
			\draw [decorate,decoration={brace,amplitude=10pt,mirror,raise=4pt}, line width=0.5pt] (9.1,-2) -- (11,-2);
			\node at (10.05, -3.5) {\LARGE $C_3$};
		\end{tikzpicture}
		\caption{A sketch of the reduction from~\Cref{DomTimelineMaxDeg1Hardness}.
Let $G$ be a graph with maximum degree at most four, where $E(G) = F_1 \cup \dots \cup F_5$ is a partition of the edges into the colors.
The blue edge in the figure exists in $G_{15+j}$ if and only if $v_i$ is isolated in the graph~$(V, F_j)$.
Otherwise, the red edge exists.
The parts $C_2$ and $C_3$ are constructed in a similar way as $C_1$, but $ u_i$ swaps roles with $u_i'$ and $u_i''$ respectively.}
		\label{Figure3ColorabilityToDomTimeline}
	\end{figure}
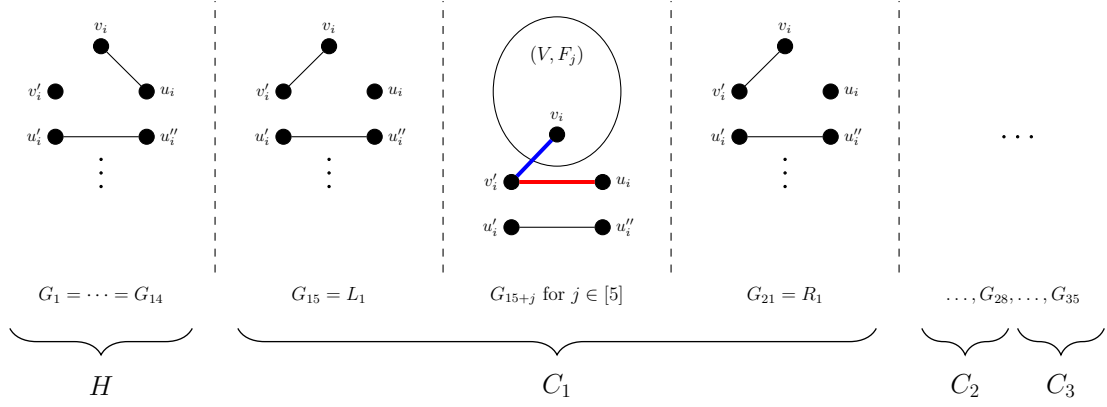

	\textit{Correctness:} We show that $G$ is $3$-colorable if and only if there is a $3$-activity $6$-timeline~$\mathcal{T}$ dominating all temporal vertices of~$\mathcal{G}$.
From now on we call the first part of $H$, consisting of the first seven snapshots, $H_1$ and we call the second part, consisting of the other seven snapshots, $H_2$.
The reason for this is simply because the activity intervals have length six (so they intersect seven snapshots).
In the following we say a vertex~$v$ in active during~$H_i$ ($i\in[2]$) or during~$C_j$ ($j\in[3]$) to indicate that an activity interval of~$v$ starts at the first time step and ends and the last time step of the corresponding part.

	($\Rightarrow$) Suppose $\phi : V(G) \rightarrow [3]$ is a proper $3$-coloring of $G$.
We construct an activity timeline~$\mathcal{T}$ dominating all temporal vertices of~$\mathcal{G}$.
For each $i \in [n]$, the vertex $v_i$ is active in the two color blocks, which do not correspond to the assigned color of $v_i$ and the vertex $v'_i$ is active in the color block, which corresponds to the assigned color of $v_i$.
All the remaining activity intervals of $v_i, v'_i, u_i, u'_i$ and $u''_i$ ensure that all of their temporal vertices are dominated.
In particular, these remaining activity intervals do not depend on the color of $v_i$.
Informally, our activity timeline $\mathcal{T}$ is constructed in the following way.
For each $i \in [n]$ with $\phi(v_i) = j$, let
	\begin{align*}
		& v_i \text{ be active in } H_1 \text{ and in } C_r \text{ for } r \neq j, \\
		& v_i' \text{ be active in } H_1, H_2, C_j, \\
		& u_i \text{ be active in } H_2, C_1, C_2, \\
		& u_i' \text{ be active in } H_1, C_2, C_3, \\
		& u_i'' \text{ be active in } H_2, C_1, C_3.
	\end{align*}
	Formally, we define for each color $j \in [3]$ the set $\mathcal{T}_j$ by
	\begin{align*}
		& \mathcal{T}_1 := \bigcup_{\substack{v_i \in V: \\ \phi(v_i) = 1}} \{ (v_i, 22, 28), (v_i, 29, 35), (v_i', 15, 21 ) \}, \\
		& \mathcal{T}_2 := \bigcup_{\substack{v_i \in V: \\ \phi(v_i) = 2}} \{ (v_i, 15, 21), (v_i, 29, 35), (v_i', 22, 28 ) \}, \\
		& \mathcal{T}_3 := \bigcup_{\substack{v_i \in V: \\ \phi(v_i) = 3}} \{ (v_i, 15, 21), (v_i, 22, 28), (v_i', 29, 35 ) \}
	\end{align*}
	and we define for each $i \in [n]$ the set
	\begin{align*}
		\mathcal{H}_i := \{ & (v_i, 1, 7), (v_i',1,7), (v_i', 8, 14), (u_i, 8,14), (u_i, 15,21), (u_i, 22, 28) \\
		& (u_i',1,7), (u_i',22,28), (u_i ', 29, 35), (u_i'', 8,14), (u_i'',15,21), (u_i'', 29,35)\}.
	\end{align*}
	This allows us to define $$\mathcal{T} := \mathcal{T}_1 \, \cup \, \mathcal{T}_2 \, \cup \, \mathcal{T}_3\,  \cup\,  \bigcup_{i = 1}^n \, \mathcal{H}_i.$$ Note that by construction $\mathcal{T}$ is a $3$-activity $6$-timeline and it remains to show that~$\mathcal{T}$ dominates all temporal vertices in $\mathcal{G}$.

	First, consider the part $H$ (the first $14$ snapshots).
The four vertices $v_i,u_i,u'_i$ and~$u''_i$ are dominated in every snapshot of $H$, because the edges~$v_iu_i, u'_iu''_i$ exist in every snapshot of $H$ and $\{(v_i,1,7), (u_i,8,14), (u'_i,1,7), (u''_i,8,14)\}\subseteq \mathcal{T}$.
The remaining vertex $v'_i$ is dominated in these snapshots, because it is active in all snapshots of them.

	Next, we consider the color part $C_1$.
The two vertices $u'_i$ and $u''_i$ are dominated, because the edge $u'_iu''_i$ exists in every snapshot of $C_1$ and $(u'_i,15,21) \in \mathcal{T}$.
The vertex $u_i$ is dominated, because it is active in all of $C_1$.
For $v_i$ there are two cases.
If $\alpha(v_i) = 1$, then $v_i$ is not active in $C_1$ and $v'_i$ is active in all of $C_1$.
Consequently,~$v'_i$ is dominated in $C_1$ and~$v_i$ is dominated in each snapshot of $C_1$ either by $v'_i$ or by a neighbor with respect to~$G$.
For this note that each neighbor $u$ of $v_i$ in $G$ is active in $C_1$, since $\alpha (u) \neq 1$.
If $\alpha (v_i) \neq 1$, then $v_i$ is active in $C_1$ and therefore dominated.
In this case the vertex $v'_i$ is not active in $C_1$ and is dominated in each snapshot of~$C_1$ either by $v_i$ or by $u_i$.
Therefore all temporal vertices in the color block $C_1$ are dominated by $\mathcal{T}$.
The same reasoning holds for the color blocks $C_2$ and $C_3$ by swapping the roles of $u_i$ with $u'_i$ and $u''_i$ respectively.

	($\Leftarrow$) For the other direction suppose $\mathcal{T}$ is a $3$-activity $6$-timeline dominating all temporal vertices of~$\mathcal{G}$.
We want to construct a valid $3$-coloring for $G$.
The idea is to assign $v_i$ the color $j$ for which $v_i$ is not active in $C_j$.
For this purpose, we need to show that two of the three activity intervals of $v_i$ correspond to exactly two of the color blocks.
We split the proof into smaller parts for more clarity.

	Since $v'_i$  is isolated in the first fourteen snapshots, we observe the following.
	\begin{observationProof}
		\label{Observation1}
		At least two activity intervals of $v_i'$ are completely contained in the part~$H$.
	\end{observationProof}
	Using~\Cref{Observation1}, we can prove the next claim.
	\begin{claim}
		\label{Claim2}
		The vertex $v_i$ has exactly one activity interval, which intersects part $H$.
This activity interval either corresponds to the time steps of~$H_1$ or~$H_2$.
	\end{claim}
	\begin{claimproof}
		First, note that in each first snapshot $L_j$ and last snapshot $R_j$ of a color block~$C_j$ ($j \in [3]$), one of the two vertices $v_i$ and $v_i'$ needs to be active.
By~\Cref{Observation1}, the vertex $v_i'$ has at most one activity interval, which intersects the color blocks.
Further, observe that each activity interval can intersect at most one snapshot $L_j$ and one snapshot $R_j$ at the same time (and this only if the activity interval does not intersect $H$), because the length of an activity interval is at most six.
This implies that $v'_i$ is active in at most one $L_j$ and in at most $R_j$.
Moreover, it follows that $v_i$ has at most one activity interval which intersects $H$, because the activity intervals from $v_i$ need to intersect in sum at least two first and two last snapshots of the color blocks and any activity interval intersecting $H$ can only intersect $L_1$ and no $R_j$.

		Now we show that at most one activity interval of $u_i$ intersects $H$.
For contradiction, assume that two activity intervals of $u_i$ intersect part $H$.
This implies that the remaining third activity interval of $u_i$ has to intersect~$R_1$, since $u_i$ is isolated in this snapshot (and no activity interval can intersect~$H$ and~$R_1$ at the same time).
Consequently, $u_i$ is not active in both~$R_2$ and $R_3$.
Since~$u_i$ is adjacent to $u_i''$ and~$u_i'$ in~$R_2$ and $R_3$ respectively, it follows that both~$u_i'$ and~$u_i''$ need to be active in both~$R_2$ and $R_3$ (once for dominating $u_i$ and once they are isolated).
No single activity interval can intersect both $R_2$ and $R_3$ and therefore at most one activity interval from each~$u_i'$ and~$u_i''$ can intersect the parts $H$ and $C_1$.
As $u'_iu''_i$ is an isolated edge in all snapshots of $H$ and $C_1$, we obtain a contradiction as not every temporal vertex of $u_i'$ and $u_i''$ can be dominated in~$H$ and $C_1$.

		Hence, we conclude that $u_i$ and $v_i$ have at most one activity interval which intersects $H$.
This directly implies that~$v_i$ is active either in exactly $H_1$ or $H_2$ (recall that $v_iu_i$ is an isolated edge in all snapshots of $H$).
	\end{claimproof}

	From~\Cref{Observation1} and~\Cref{Claim2} we conclude that there are in total three activity intervals of~$v_i$ and $v_i'$ (two of them are from $v_i$), which intersect the color blocks.
As mentioned above, in each first snapshot $L_j$ and in each last snapshot $R_j$ of a color block $C_j$, one of the vertices $v_i$ or $v'_i$ needs to be active.
Also recall that any activity interval intersects at most one snapshot $L_j$ and at most one snapshot~$R_j$ at the same time.
Consequently, each of these three remaining activity intervals needs to intersect a different $L_j$ and different $R_j$ than the other ones.
This directly implies that the three remaining activity intervals of $v_i$ and $v_i'$ occupy exactly the three color blocks $C_1, C_2$ and $C_3$.

	We can define a $3$-coloring of $G$ by assigning $v_i$ the color $j$ if and only if $v_i$ is not active in the color block~$C_j$.
If $v_i$ is not active in $C_j$, then by construction, each $w \in N_G (v_i)$ has to be active in~$C_j$ (recall that $v_iw$ is an edge in some of the snapshots of $C_j$ and the maximum snapshot degree is one, so $v_i$ has to be dominated by $w$ in some snapshot of the color block $C_j$) and therefore is assigned a different color than~$v_i$.
\end{proof}

\section{The Influence of Membership-Width Based Parameters}
\label{sec-membership}

Now, we investigate the membership-width based parameters~$\vimw$ and~$\imw$.
More precisely, we study the parameter combinations~$\vimw +\,k+\ell$ and~$\imw +\,k+\ell$. These combinations are motivated by the fact that \CT and \DT are W[1]-hard with respect to~$n+\ell$; for \CT this follows from previous work~\cite{DBLP:journals/mst/FroeseKZ24} and for \DT, we will show this in \Cref{sec-input-paras}.

\subsection{The Parameter Vertex-Interval-Membership-Width}

In this section, we provide FPT-algorithms with respect to~$\vimw + \,k+\ell$ for \PCT and \PDT.
Moreover, we show that~$\vimw$ can be replaced by smaller parameters~$\vimw[x]$ for all problems except \PDT (the precise value of~$x$ differs for \PCT and \DT).
Simultaneously, these algorithms also are XP-algorithms for \PCT and \PDT parameterized by~$\vimw$ alone.
Initially, we focus on \PCTT and we show that \PCT admits an FPT-algorithm for~$\vimw + \,k+\ell$.

\begin{theorem}
	
	\label{CovTimeLine_vimw+k+l}
	\PCT is solvable in $\big((k+1)(\ell + 2)\big)^{2\cdot \vimw} \cdot (n+T)^{\Oh(1)}$~time.
\end{theorem}
\begin{proof}
	Let $(\mathcal{G} = (G_1, \dots, G_T),k, \ell, t)$ be an instance of \PCT.
Further, suppose that $T \geq k \cdot (\ell + 1)+1$ (otherwise, it is a trivial instance) and that the underlying graph $G_\downarrow$ contains no isolated vertex~$v$.
Otherwise, since~$v$ does not cover any temporal edge, $v$ can be safely removed and consequently $t$ remains unchanged.
Moreover, we assume that each vertex~$v$ is contained in at least $k\cdot(\ell+1)+1$~bags (recall that these are consecutive) because otherwise we can simply greedily pick $k$ activity intervals of~$v$ which contain all snapshots where~$v$ is non-isolated and thus we cover all temporal edges incident with~$v$.
Thus, it is safe to remove~$v$ from~$\mathcal{G}$ and reduce~$t$ by the number of temporal edges incident with~$v$ over all snapshots.

	The idea is to use dynamic programming over the lifetime of the temporal graph.
	Each table entry corresponds to the maximal number of covered temporal edges where a specific set of vertices in the bag is active.
In other words, the dynamic program keeps track of the activity of the vertices which are contained in the bag $F_i$ of the currently considered time step $i \in [T]$.
While iterating over the lifetime of the temporal graph, we need to ensure that the activities at a time step are compatible with the activities at the neighboring time steps.
Moreover, for a vertex~$v$ let~$i$ and~$j$ be the index of the smallest and largest snapshot where~$v$ is incident to an edge, respectively.
By definition, $v\in F_x$ for any~$i\le x\le j$.
Furthermore, since~$v$ is isolated in all snapshots~$G_x$ where~$x<i$ or~$x>j$, it is safe to assume that all $k$~activity intervals of~$v$ are used when~$v$ is contained in the bags~$F_x$.

	Recall that the vertex-interval-membership sequence can be computed in polynomial time. Since we assumed that the underlying graph contains no isolated vertex, it follows that each vertex is contained in at least one bag of the sequence.

	\emph{Table Definition:}
	We denote $\overleftarrow{F}_i := F_1 \cup \dots \cup F_i$ and $\overline{F}_i := \overleftarrow{F}_i \setminus F_i$.
In other words, $\overline{F}_i$ is the set of vertices which are isolated in all snapshots~$G_i, \dots, G_T$ since we assumed that~$G_\downarrow$ contains no isolated vertices.
Moreover, for $i \in [T]$ we define functions~${\curr : F_i \rightarrow [0,k],}$ and $\pos : F_i \rightarrow [0, \ell + 1]$.
The functions~$\curr$ and~$\pos$ indicate for each vertex~$v\in F_i$ whether an activity interval of~$v$ is active during snapshot~$G_i$.
More precisely, the function value $\curr(v)$ determines the number of the current activity interval of $v$ and the function value~$\pos(v)$ determines the position in this activity interval.
Here, $\curr (v) = 0$ means that the current time step is before the first activity interval and $\pos(v) = 0$ means that $v$ is currently not active, that is, the $\curr(v)$-th activity interval of $v$ is already over and the next one has not started yet.
Now the entry $\DP[i, \curr, \pos]$ contains the maximum number of covered temporal edges in the first~$i$ snapshots~$G_1, \dots , G_i$, while $i$ is at the~$\pos (v)$-th position of the $\curr(v)$-th activity interval of $v \in F_i$.

The formal definition of the table now is
	\begin{align*}
		\DP[i, \curr, \pos] \, \hat{=} \, \max_{\mathcal{T}}  \, & |\{(e, j) \in E_\downarrow \times [i] : (e, j) \text{ is covered by } \mathcal{T}\}|,
	\end{align*}
	where the maximum is taken over all $k$-activity $\ell$-timelines $$\mathcal{T} \subseteq (\overline{F}_i \times [i-1] \times [T]) \cup (F_i \times [i] \times [T]),$$ such that for each $v \in F_i$
	\begin{enumerate}[label=(\roman*)]
		\item there are at most $\curr(v)$ activity intervals of $v$ in $\mathcal{T}$,
		\item $\pos (v) \leq i$, and
		\item $(v, i - \pos(v) + 1, b) \in \mathcal{T}$ for~$b=\min(i - \pos(v) + 1 + \ell, T)$ if~$\pos (v) > 0$.
	\end{enumerate}

	We need Conditions~(i)-(iii) for the following reasons:
	Condition~(i) ensures that each vertex has at most~$k$ activity intervals.
Moreover, condition (ii) ensures that no activity interval starts before time step one.
Finally, condition (iii) ensures that each activity interval has length~$\ell$, except if the remaining number of snapshots is too small.

\emph{Initialization:}
	For all $\curr : F_1 \rightarrow [0,k]$ and $\pos : F_1 \rightarrow [0, 1]$, for which $\curr(v) = 0$ implies $\pos (v) = 0$, the table is initialized by

	$$\DP [1, \curr, \pos] = |E_{G_1}(\{v \in F_1: \pos(v) = 1\}, F_1)|.$$

	Note that $\overline{F}_1 = \emptyset$ and therefore only activity timelines $\mathcal{T} \subseteq F_1 \times [1] \times [T]$ are considered in the definition of $\DP [1, \curr, \pos]$.
The initialization is correct, since the table entry~$\DP [1, \curr, \pos]$ contains the number of temporal edges which are covered by the active vertices from $F_1$ in~$G_1$.
Observe that the other endpoints of the covered edges are also contained in $F_1$ by definition.

	\emph{Recurrence:}
	The remaining table entries are computed recursively.
Intuitively, $\DP[i, \curr, \pos]$ is computed by trying all activity interval choices for the vertices of $F_{i-1}$, such that there is no conflict with $\curr$ and $\pos$.
	Formally, we set
	\begin{align*}
		\DP[i, \curr, \pos] = \max_{\curr', \pos'} \, &\DP[i-1, \curr', \pos']
		 + |E_{G_i}(\{v \in F_i: \pos(v) \geq 1\},F_i) |
	\end{align*}

	where the maximum is taken over all $\curr' : F_{i-1} \rightarrow [0,k]$ and $\pos' : F_{i-1} \rightarrow [0, \ell + 1]$ which are \emph{compatible} with $\curr$ and $\pos$.
Informally, $\curr'$ and $\pos'$ are \emph{compatible} with $\curr$ and $\pos$ if they provide a correct extension of the activities of $F_{i-1} \cap F_i$ at time step~$i$ to time step~$i-1$.
Formally, four~conditions need to be fulfilled:
First, if an interval of~$v$ is already active at time step~$i$, that is, if~$\pos(v)\ge 2$, then an activity interval of~$v$ is already active during time step~$i-1$.
This can be formalized as follows.
\begin{equation}
  \label{eq:cond1}
          \text{If } \pos(v) \geq 2\text{, then } \curr' (v) \leq \curr(v) \text{ and } \pos' (v) = \pos (v) - 1.
\end{equation}
Second, if a new activity interval of~$v$ starts at time step~$i$ either~$v$ was inactive at time step~$i-1$ or the previous activity interval of~$v$ ended at time step~$i-1$.
This can be formalized as follows.
\begin{equation}
  \label{eq:cond2}
          \text{If } \curr(v) > 1 \text{ and } \pos (v) = 1\text{, then } \curr' (v) \leq \curr (v) - 1 \text{ and } \pos' (v) \in \{\min (i-1, \ell + 1),0\}.
\end{equation}
Third, if the first activity interval of~$v$ starts at time step~$i$, then in the previous time step~$i-1$ vertex~$v$ cannot be active. This can be formalized as follows.
\begin{equation}
  \label{eq:cond3}
\text{If } \curr(v) = 1 \text{ and } \pos (v) = 1\text{, then } \curr' (v) = 0 \text{ and } \pos' (v) = 0.
\end{equation}
Finally, if~$v$ is not active during time step~$i$, either~$v$ is also not active during time step~$i-1$ or an activity interval of~$v$ ended at time step~$i-1$.
This can be formalized as follows.
\begin{equation}
  \label{eq:cond4}
  \text{If } \curr(v) \geq 1 \text{ and } \pos (v) = 0\text{, then } \curr' (v) \leq \curr (v) \text{ and } \pos' (v) \in \{ \min (i-1, \ell + 1),0\} .
\end{equation}

	The minimum in Conditions~\ref{eq:cond2} and~\ref{eq:cond4} ensures that the activity intervals do not have a starting point smaller than one.

	\emph{Correctness:}
	We show the correctness of the computation by proving inequalities in both directions.

	($\leq$) Suppose the maximum in the definition of the table entry $\DP[i, \curr, \pos]$ is attained for~$\mathcal{T}$ and assume that no activity intervals from $\mathcal{T}$ overlap.
We define $\mathcal{T}'$ by taking all activity intervals from $\mathcal{T}$ which are relevant for a table entry of $i-1$.
Formally, we define
	\begin{align*}
		\mathcal{T}' := & \{(v,a,b) \in \mathcal{T} : a \leq i-2, v \in \overline{F}_{i-1}\} \\
			& \cup \{(v,a,b) \in \mathcal{T} : a \leq i - 1, v \in F_{i-1}\},
		\end{align*}
		and, in accordance with $\mathcal{T}'$, functions $\curr' : F_{i-1} \rightarrow [0,k], \pos' : F_{i-1} \rightarrow [0, \ell + 1]$ by
		\begin{align*}
			\curr'(v) &:= |\{(v,a,b) \in \mathcal{T}' : a,b \in [T]\}| \\
			\pos'(v) &:=
			\begin{cases}
				i-a & \text{if } (v,a,b) \in \mathcal{T}'  \text{ with } a \leq i \leq b\\
				0 & \text{otherwise.}
			\end{cases}
		\end{align*}
	Note that $\mathcal{T}' \subseteq \mathcal{T}$.
We show that $\curr'$ and $\pos'$ satisfy Conditions (1)--(4).
This means that the table entry $\DP[i-1, \curr', \pos']$ is considered in the maximum on the right side of the recursive computation.
By definition of $\curr'$ and $\pos'$ it is clear that $\mathcal{T}'$ is then considered in the definition of the table entry $\DP[i-1, \curr', \pos']$.
Now let $v \in F_{i-1} \cap F_i$.
	\begin{enumerate}[label=(\arabic*)]
		\item If $\pos(v) \geq 2$, then $(v,a,b) \in \mathcal{T}$ for some $a \leq i-1, b \in [T]$ and by definition of~$\mathcal{T}'$, each such activity interval of $v$ in $\mathcal{T}$ is also contained in $\mathcal{T}'$.
Consequently, we have $\curr'(v) \leq \curr(v)$ and $\pos'(v) = \pos (v) - 1$.
		\item If $\curr(v) > 1$ and $\pos (v) = 1$, then $(v,i,b) \in \mathcal{T} \setminus \mathcal{T}'$ for some $b \in [T]$.
So the activity timeline~$\mathcal{T}'$ contains at most $\curr(v)-1$ activity intervals of $v$ and therefore $\curr'(v) \leq \curr(v) - 1$.
Now either $(v,a,i-1) \in \mathcal{T}'$ for some $a \in [i-1]$ or $v$ is not active in $G_{i-1}$ with respect to the activity intervals from $\mathcal{T}'$.
This means that either $\pos' (v) = \min (i-1, \ell + 1)$ or $\pos' (v) = 0$.
		\item If $\curr(v) = 1$ and $\pos(v) =1$, then $(v,i,b) \in \mathcal{T} \setminus \mathcal{T}'$ for some $b \in [T]$.
Consequently, $\mathcal{T}'$~contains no activity interval of $v$ and therefore $\curr' (v) = \pos' (v) = 0$.
		\item If $\curr (v) \geq 1$ and $\pos (v) = 0$, then either $(v,a,i-1) \in \mathcal{T}'$ for some $a \in [i-1]$ or $v$ is not active in $G_{i-1}$ with respect to the activity intervals from $\mathcal{T}'$.
So by definition $\curr' (v) \leq \curr(v)$ and either $\pos' (v) = \min (i-1, \ell + 1)$ or $\pos' (v) = 0$.
	\end{enumerate}
	Now consider the set $\mathcal{A}$ of temporal edges until time step~$i$ that are covered by $\mathcal{T}$ and the set $\mathcal{A'}$ of temporal edges until time step~$i-1$ that are covered by $\mathcal{T'}$.
	        There are two cases for a temporal edge~$(vw,j)\in \mathcal{A} \setminus \mathcal{A}'$.

\emph{Case 1:  $j = i$.} Since the temporal edge~$(vw,j)$ is covered by~$\mathcal{T}$ we can conclude that~$vw \in E_{G_i}(\{v \in F_i: \pos(v) \geq 1\}, F_i)$.

\emph{Case~2: $j < i$.} We have $(vw,j)\in \binom{\overleftarrow{F}_{i-1}}{2}  \times [i-1]$.
This implies the existence of an activity interval~$(v,a,b)$ or~$(w,a,b)$ in~$\mathcal{T}$ with~$a \leq i-1$.
However, then the activity interval is by definition contained in~$\mathcal{T}'$ and consequently~$(vw,j) \in \mathcal{A}'$.

Hence, any temporal edge in $\mathcal{A} \setminus \mathcal{A}'$ is contained in~$E_{G_i}(\{v \in F_i: \pos(v) \geq 1\}, F_i)$. We thus have the desired inequality 
\begin{align*}
  \DP [i,\curr,\pos] = |\mathcal{A}| & \leq |\mathcal{A}'| + |E_{G_i}(\{v \in F_i: \pos(v) \geq 1\},F_i) |\\ & \leq  \DP[i-1, \curr', \pos'] + |E_{G_i}(\{v \in F_i: \pos(v) \geq 1\},F_i) |.
\end{align*}

	($\geq$) Suppose the maximum on the right side of the computation is attained for~$\curr',\pos'$ and the maximum in the definition of $\DP[i-1, \curr', \pos']$ is attained for $\mathcal{T}'$.
We define $\mathcal{T}$ by extending $\mathcal{T}'$ in the following way:
For $v \in F_{i-1} \cap F_i$ the activity interval $(v,i,\min(i+ \ell, T))$ is added if and only if $\pos (v) = 1$. These new intervals together with the already active ones ($\pos(v) \geq 2$) cover all edges of~$E_{G_i}(\{v \in F_i: \pos(v) \geq 1\},F_i)$. Hence, the total number of edges covered by $\mathcal{T}$ is $\DP[i-1, \curr', \pos'] + |E_{G_i}(\{v \in F_i: \pos(v) \geq 1\},F_i)|$.

Because $\mathcal{T}$ is considered in the definition of $\DP[i, \curr, \pos]$, we have $$\DP[i, \curr, \pos] \geq \DP[i-1, \curr', \pos'] + |E_{G_i}(\{v \in F_i: \pos(v) \geq 1\},F_i)|.$$
Hence, the desired inequality holds.

	\emph{Result:}
	Finally, we return yes if and only if $\DP[T, \curr, \pos] \geq t$ for $\curr : F_T \rightarrow \{k\}$ and some function $\pos : F_T \rightarrow \{0, \ell + 1\}$.
The restriction to these functions is correct, since we can assume without loss of generality that a vertex is active in $G_T$ if and only if its~$k$-th activity interval ends at time step $T$ (recall that we assume that each vertex is in at least $k(\ell+1)+1$~bags).
Consequently, all $k$-activity $\ell$-timelines of $\mathcal{G}$ are considered in the definition of the table entry $\DP[T, \curr, \pos]$ and all covered temporal edges of $\mathcal{G}$ are counted in the respective maximum of the definition.

	\emph{Running Time:}
	For each time step $i \in [T]$, there are $\big((k+1)(\ell + 2)\big)^{\vimw}$ choices for the functions $\curr : F_i \rightarrow [0,k]$ and $\pos : F_i \rightarrow [0, \ell + 1]$.
This upper-bounds the size of our dynamic programming table by $T \cdot \big((k+1)(\ell + 2)\big)^{\vimw}$.
In each recursive computation of a table entry, there are also $\Oh( \big((k+1)(\ell + 2)\big)^{\vimw})$ choices for $\curr'$ and $\pos'$.
The remaining summands of the recursive formula can be computed in~$\Oh(n)$ time.
Hence, the overall running time can be bounded by~$\big((k+1)(\ell + 2)\big)^{2\cdot \vimw} \cdot (n+T)^{\Oh(1)}$.
\end{proof}

Now, we show that for \PCT the parameter~$\vimw=\vimw[1]$ can be replaced by an even smaller parameter,~$\vimw[k(\ell+1)+1]$.
The algorithm for this parameter has two steps: First it applies preprocessing step that handles the large bags. Then, it invokes the algorithm of \Cref{CovTimeLine_vimw+k+l} for~$\vimw=\vimw[1]$.

\begin{proposition}
	\label{CovTimeLine_vimw+k+l2}
	\PCT is solvable in $\big((k+1)(\ell + 2)\big)^{4\cdot \vimw[x]} \cdot (n+T)^{\Oh(1)}$~time where~$x=k(\ell+1)+1$.
\end{proposition}
\begin{proof}
Let~$(F_i)_{i\in [T]}$ be the bags of the vertex-interval-membership sequence of~$\mathcal{G}$.
Similar to the proof of \Cref{CovTimeLine_vimw+k+l}, it is safe to assume that each vertex~$v$ is contained in at least $k\cdot (\ell+1)$~bags (by definition the bags containing~$v$ need to be consecutive) because otherwise, we can greedily pick the $k$-activity intervals of~$v$ such that~$v$ is active during all these snapshots and thus all temporal edges incident to~$v$ are covered and we can reduce~$t$ accordingly.

We call the $k\cdot(\ell+1)$~largest bags of~$(F_i)_{i\in [T]}$ \emph{large}.
The idea is to pick the intervals greedily for all vertices that are only contained in large bags.
Intuitively, this works since the large bags appear consecutively.
Afterwards, we  remove all these vertices and adapt~$t$ accordingly.
Finally, we can invoke the algorithm of \Cref{CovTimeLine_vimw+k+l}.

Now, let~$[i,j]\subset [T]$ be a sequence of consecutive indices such that each bag~$F_x$ for each~$x\in[i,j]$ is large.
By definition~$j-i\le k\cdot(\ell+1) - 1$.
Now, let~$\mathcal{A}= \bigcup_{x\in[i,j]}F_x\setminus (F_{i-1}\cup F_{j+1})$.
If~$i=0$ or~$j=T$, then~$F_{i-1}$ or~$F_{j+1}$ is simply the empty set.
Suppose~$\mathcal{A} \ne \emptyset$.
By definition, each vertex in~$\mathcal{A}$ is isolated in~$G_1, \ldots, G_{i-1}$ and also in~$G_{j+1}, \ldots, G_T$.
By our assumption that each vertex is contained in at least $k\cdot(\ell+1)$~bags, we can conclude that~$j-i=k\cdot(\ell+1)-1$.
Now observe that by picking the $k$-activity intervals for each vertex in~$\mathcal{A}$ greedily, we can cover all temporal edges which are incident to some vertex of~$\mathcal{A}$.
Hence, it is safe to remove~$\mathcal{A}$ from~$\mathcal{G}$ and to reduce~$t$ by the number of temporal edges covered by the vertices in~$\mathcal{A}$.
For the remaining instance, we use the algorithm of \Cref{CovTimeLine_vimw+k+l}.
After removing~$\mathcal{A}$, we have~$|F_x| \le |F_{i-1}| + |F_{j+1}| \le 2 \vimw [k (\ell + 1) + 1]$ for every~$x \in [i,j]$ and consequently~\Cref{CovTimeLine_vimw+k+l} yields the desired running time.
Now, suppose~$\mathcal{A} = \emptyset$ and let~$F_x$ be a large bag contained in a sequence of large bags with indices~$[i,j]$.
In this case, we can conclude~$|F_x| \le |F_{i-1}| + |F_{j+1}| \le 2 \vimw [k (\ell + 1) + 1]$ and again apply~\Cref{CovTimeLine_vimw+k+l}.
\end{proof}

Now, we focus on \PDTT.
Initially, we show how to adapt the algorithm of \Cref{CovTimeLine_vimw+k+l} for \PDT to obtain an algorithm with the same running time for \PDT. However, we need to be more careful since isolated vertices in the snapshots can no longer be ignored as for \PCT.
This means that vertices which are not contained in the current bag of the vertex-interval-membership sequence might have to be active.
Hence, the update of the dynamic program additionally needs to consider the case that activity intervals of vertices appear before or after their first or last occurrence in a bag.

\begin{theorem}

	\label{DomTimeLine_vimw+k+l}
	\PDT is solvable in $\big((k+1)(\ell + 2)\big)^{2\cdot \vimw} \cdot (n+T)^{\Oh(1)}$ time.
\end{theorem}

\begin{proof}
	Let $(\mathcal{G} = (G_1, \dots, G_T),k, \ell, t)$ be an instance of \PDT.
Again, suppose that $T \geq k \cdot (\ell + 1)$ (otherwise, it is a trivial instance) and that the underlying graph $G_\downarrow$ contains no isolated vertex (otherwise, decrease $t$ by $k \cdot (\ell + 1)$ and delete the vertex).

	The idea of the dynamic program over the lifetime of the temporal graph is analogous to the algorithm of \Cref{CovTimeLine_vimw+k+l}.
Recall that the dynamic program keeps track of the activity of the vertices which are contained in the bag $F_i$ of the currently considered time step $i \in [T]$.
Again, while iterating over the lifetime of the temporal graph, we need to ensure that the activities at a time step are compatible with the activities at the neighboring time steps.
Moreover, we need to be careful as it might be optimal for some vertices to be active at time steps where they are not contained in the respective bag to dominate some isolated vertices.

	Recall that the vertex-interval-membership sequence can be computed in polynomial time.
Since we assumed that the underlying graph contains no isolated vertex, it follows that each vertex is contained in at least one bag of the sequence.

	\emph{Table Definition:}
	Again we let $\overleftarrow{F}_i := F_1 \cup \dots \cup F_i$ and $\overline{F}_i := \overleftarrow{F}_i \setminus F_i$.
Recall that~$\overline{F}_i$ is the set of vertices which are isolated in all snapshots $G_i, \dots, G_T$.
Moreover, again for $i \in [T]$ we define functions~${\curr : F_i \rightarrow [0,k],}$ and $\pos : F_i \rightarrow [0, \ell + 1]$ which indicate for each vertex~$v\in F_i$ whether an activity interval of~$v$ is active during snapshot~$G_i$.
Recall that the function value~$\curr(v)$ determines the current activity interval of $v$ and the function value $\pos(v)$ determines the position in this activity interval.

Now, the entry $\DP[i, \curr, \pos]$ contains the combined maximum number of dominated temporal vertices of~$\overline{F}_i$ in the whole temporal graph and of $F_i$ in the snapshots~$G_1, \dots , G_i$, while $i$ is at the~$\pos (v)$-th position of the $\curr(v)$-th activity interval of $v \in F_i$.
We are now ready to present the formal definition of the table.
	\begin{align*}
		\DP[i, \curr, \pos] \, \hat{=} \, \max_{\mathcal{T}}  \, & |\{(v, j) \in \overline{F}_i \times [T] : (v, j) \text{ is dominated by } \mathcal{T}\}\\
		& \cup \{(v, j) \in F_i \times [i] : (v, j) \text{ is dominated by } \mathcal{T}\}|,
	\end{align*}
	where the maximum is taken over all $k$-activity $\ell$-timelines $$\mathcal{T} \subseteq (\overline{F}_i \times [T] \times [T]) \cup (F_i \times [i] \times [T]),$$ such that for each $v \in F_i$
	\begin{enumerate}[label=(\roman*)]
		\item there are at most $\curr(v)$ activity intervals of $v$ in $\mathcal{T}$,
		\item $\pos (v) \leq i$, and
		\item $(v, i - \pos(v) + 1, b) \in \mathcal{T}$ for some $b \geq 1$ if~$\pos (v) > 0$.
	\end{enumerate}

We need conditions~(i)-(iii) for the identical reasons as for \PCT:
	Condition~(i) ensures that each vertex has at most~$k$ activity intervals.
Moreover, condition (ii) ensures that no activity interval starts before time step one.
Finally, condition (iii) ensures that each activity interval has length~$\ell$, except if the remaining number of snapshots is too small.

	\emph{Initialization:}
	For all $\curr : F_1 \rightarrow [0,k]$ and $\pos : F_1 \rightarrow [0, 1]$, for which $\curr(v) = 0$ implies $\pos (v) = 0$, the table is initialized by
	$$\DP [1, \curr, \pos] = |N_{G_1} [\{v \in F_1: \pos(v) = 1\} ]|.$$
	Note that $\overline{F}_1 = \emptyset$ and therefore only activity timelines $\mathcal{T} \subseteq F_1 \times [1] \times [T]$ are considered in the definition of $\DP [1, \curr, \pos]$.
The initialization is correct, since the table entry $\DP [1, \curr, \pos]$ contains the number of vertices which are dominated by active vertices from $F_1$ in $G_1$ and the dominated vertices are also contained in $F_1$ by definition.

	\emph{Recurrence:}
Intuitively, $\DP[i, \curr, \pos]$ is computed by trying all activity interval choices for the vertices of $F_{i-1}$, such that there is no conflict with $\curr$ and $\pos$.
Clearly, all vertices which are either active during time step~$i$ or a neighbor of an active vertex in snapshot~$G_i$ get dominated (summand~(A1)).
However, we also need to consider how many of the vertices from $F_{i}$ and $F_{i-1}$ are isolated in $G_{1}, \dots , G_{i-1}$ and $G_i, \dots, G_{T}$ respectively.
In particular, we need to count how many of these isolated temporal vertices are dominated by a solution candidate.
More precisely, two cases are possible:
First, if~$v \in F_i$ is isolated in~$G_1, \ldots, G_{i-1}$, then in each time step in~$[i-1]$, which intersects an activity interval of~$v$, exactly one vertex, that is, $v$ itself is dominated (summands~(A2) and~(A3)).
Second, if~$v \in F_{i}$ is isolated in~$G_{i+1}, \ldots, G_{T}$, then in each time step in~$[i+1,T]$, which intersects an activity interval of~$v$, exactly one vertex, that is, $v$ itself is dominated (summands~(A4) and~(A5)).
Consequently, the formal update is as follows:
	\begin{align*}
		\DP[i, \curr, \pos] =\\ \max_{\curr', \pos'} \, &\DP[i-1, \curr', \pos']\\
		& + |N_{G_i} [\{v \in F_i: \pos(v) \geq 1\} ]| && (A1)\\
		& + \sum_{ \substack{v \in F_i \setminus F_{i-1}: \\ \pos(v) = 0}} \min \big(\curr(v) (\ell + 1), i-1\big) && (A2) \\
		& + \sum_{ \substack{v \in F_i \setminus F_{i-1}: \\ \pos(v) > 0}} \min \big((\curr(v) - 1) (\ell + 1) + \pos(v) - 1, i-1\big) && (A3)\\
		& + \sum_{ \substack{v \in F_{i-1} \setminus F_{i}: \\ \pos'(v) = 0}} \min( (k - \curr'(v)) (\ell + 1), T - i + 1) && (A4)\\
		& + \sum_{ \substack{v \in F_{i-1} \setminus F_{i}: \\ \pos'(v) > 0}} \min( (k - \curr'(v)) (\ell + 1) + \ell + 1- \pos'(v)), T - i + 1) &&(A5),
	\end{align*}
	where the maximum is taken over all $\curr' : F_{i-1} \rightarrow [0,k]$ and $\pos' : F_{i-1} \rightarrow [0, \ell + 1]$ which are \emph{compatible} with $\curr$ and $\pos$.
	Here, compatibility has the identical meaning as in the algorithm in \Cref{CovTimeLine_vimw+k+l} for \PCT:
Informally, $\curr'$ and $\pos'$ are compatible with $\curr$ and $\pos$ if they provide a correct extension of the activities of $F_{i-1} \cap F_i$ at time step~$i$ to time step $i-1$.
Formally, four~conditions need to be fulfilled:
First, if an interval of~$v$ is already active at time step~$i$, that is, if~$\pos(v)\ge 2$, then an activity interval of~$v$ is already active during time step~$i-1$.
This can be formalized as follows.
\setcounter{equation}{0}
\begin{equation}
  \label{eq:cond1DS}
          \text{If } \pos(v) \geq 2\text{, then } \curr' (v) \leq \curr(v) \text{ and } \pos' (v) = \pos (v) - 1.
\end{equation}
Second, if a new activity interval of~$v$ starts at time step~$i$ either~$v$ was inactive at time step~$i-1$ or the previous activity interval of~$v$ ended at time step~$i-1$.
This can be formalized as follows.
\begin{equation}
  \label{eq:cond2DS}
          \text{If } \curr(v) > 1 \text{ and } \pos (v) = 1\text{, then } \curr' (v) \leq \curr (v) - 1 \text{ and } \pos' (v) \in \{\min (i-1, \ell + 1),0\}.
\end{equation}
Third, if the first activity interval of~$v$ starts at time step~$i$, then in the previous time step~$i-1$ vertex~$v$ cannot be active. This can be formalized as follows.
\begin{equation}
  \label{eq:cond3DS}
\text{If } \curr(v) = 1 \text{ and } \pos (v) = 1\text{, then } \curr' (v) = 0 \text{ and } \pos' (v) = 0.
\end{equation}
Finally, if~$v$ is not active during time step~$i$, either~$v$ is also not active during time step~$i-1$ or an activity interval of~$v$ ended at time step~$i-1$.
This can be formalized as follows.
\begin{equation}
  \label{eq:cond4DS}
  \text{If } \curr(v) \geq 1 \text{ and } \pos (v) = 0\text{, then } \curr' (v) \leq \curr (v) \text{ and } \pos' (v) \in \{ \min (i-1, \ell + 1),0\} .
\end{equation}
	The minimum in two of the conditions ensures that the activity intervals do not have a starting point smaller than one.

	\emph{Correctness:}
	($\leq$) Suppose the maximum in the definition of the table entry $\DP[i, \curr, \pos]$ is attained for $\mathcal{T}$ and assume that no activity intervals from $\mathcal{T}$ overlap.
We define $\mathcal{T}'$ by taking all activity intervals from $\mathcal{T}$, which are relevant for a table entry of $i-1$.
Formally, we define
	\begin{align*}
		\mathcal{T}' := & \{(v,a,b) \in \mathcal{T} : v \in \overline{F}_{i-1}\} \\
		& \cup \{(v,a,b) \in \mathcal{T} : a \leq i-1, v \in F_{i-1}\},
		\end{align*}
		and, in accordance with $\mathcal{T}'$, functions $\curr' : F_{i-1} \rightarrow [0,k], \pos' : F_{i-1} \rightarrow [0, \ell + 1]$ by
		\begin{align*}
			\curr'(v) &:= |\{(v,a,b) \in \mathcal{T}' : a,b \in [T]\}| \\
			\pos'(v) &:=
			\begin{cases}
				i-a & \text{if } (v,a,b) \in \mathcal{T}'  \text{ with } a \leq i \leq b\\
				0 & \text{otherwise.}
			\end{cases}
		\end{align*}
	Note that $\mathcal{T}' \subseteq \mathcal{T}$.
We show that $\curr'$ and $\pos'$ satisfy conditions (1)-(4).
This means that the table entry $\DP[i-1, \curr', \pos']$ is considered in the maximum on the right side of the recursive computation.
By definition of $\curr'$ and $\pos'$ it is clear that $\mathcal{T}'$ is then considered in the definition of the table entry $\DP[i-1, \curr', \pos']$.
Now let $v \in F_{i-1} \cap F_i$.
	\begin{enumerate}[label=(\arabic*)]
		\item If $\pos(v) \geq 2$, then $(v,a,b) \in \mathcal{T}$ for some $a \leq i-1, b \in [T]$ and by definition of~$\mathcal{T}'$, each such activity interval of $v$ in $\mathcal{T}$ is also contained in $\mathcal{T}'$.
Consequently, we have $\curr'(v) \leq \curr(v)$ and $\pos'(v) = \pos (v) - 1$.
		\item If $\curr(v) > 1$ and $\pos (v) = 1$, then $(v,i,b) \in \mathcal{T} \setminus \mathcal{T}'$ for some $b \in [T]$.
So the activity timeline $\mathcal{T}'$ contains at most $\curr(v)-1$ activity intervals of $v$ and therefore $\curr'(v) \leq \curr(v) - 1$.
Now either $(v,a,i-1) \in \mathcal{T}'$ for some $a \in [i-1]$ or $v$ is not active in $G_{i-1}$ with respect to the activity intervals from $\mathcal{T}'$.
This means that either $\pos' (v) = \min (i-1, \ell + 1)$ or $\pos' (v) = 0$.
		\item If $\curr(v) = 1$ and $\pos(v) =1$, then $(v,i,b) \in \mathcal{T} \setminus \mathcal{T}'$ for some $b \in [T]$.
Consequently, $\mathcal{T}'$ contains no activity interval of $v$ and therefore $\curr' (v) = \pos' (v) = 0$.
		\item If $\curr (v) \geq 1$ and $\pos (v) = 0$, then either $(v,a,i-1) \in \mathcal{T}'$ for some $a \in [i-1]$ or $v$ is not active in $G_{i-1}$ with respect to the activity intervals from $\mathcal{T}'$.
So by definition $\curr' (v) \leq \curr(v)$ and either $\pos' (v) = \min (i-1, \ell + 1)$ or $\pos' (v) = 0$.
	\end{enumerate}
	It remains to compare the set
	\begin{align*}
		& \mathcal{A} := \{(v, j) \in \overline{F}_i \times [T] : (v, j) \text{ is dominated by } \mathcal{T}\}\\
		& \cup \{(v, j) \in F_i \times [i] : (v, j) \text{ is dominated by } \mathcal{T}\}
	\end{align*}
	with the set
	\begin{align*}
		& \mathcal{A}':= \{(v, j) \in \overline{F}_{i-1} \times [T] : (v, j) \text{ is dominated by } \mathcal{T}'\}\\
		& \cup \{(v, j) \in F_{i-1}\times [i-1] : (v, j) \text{ is dominated by } \mathcal{T}'\}.
	\end{align*}
        
	Since $\mathcal{T}' \subseteq \mathcal{T}, \, \overline{F}_{i-1} \subseteq \overline{F}_i$ and $F_{i-1} \subseteq F_i \cup \overline{F}_i$, it is clear that $\mathcal{A}' \subseteq \mathcal{A}$.
We distinguish between three possible cases, when temporal vertices could be contained in $\mathcal{A} \setminus \mathcal{A}'$:
\textbf{a) } $(u,j)\in F_i\times \{i\}$, \textbf{b) } $(u,j)\in F_i \times [i-1]$, and \textbf{c) }  $(u,j)\in \overline{F_i} \times [T]$.
These cases are analog to the proof of \Cref{CovTimeLine_vimw+k+l} for \PCT; but now they are more involved since also isolated vertices are important for domination.

	\begin{enumerate}[label=\alph*)]
		\item If $(u,j) \in F_i \times \{i\}$ is contained in $\mathcal{A}$, then $u \in N_{G_i} [\{v \in F_i: \pos(v) \geq 1\} ]$ and therefore either $u \in F_{i-1} \cap F_i$ or $u \in F_i \setminus F_{i-1}$.
If $u \in F_{i-1} \cap F_i$, then we have~$(u, i) \notin \mathcal{A}'$ because $\mathcal{A}'$ only contains temporal vertices of $F_{i-1}$ from the first $i-1$ snapshots.
Otherwise, if $u \in F_i \setminus F_{i-1}$, then $u \notin \overline{F}_{i-1}$ and subsequently $(u,i) \notin \mathcal{A}'$.
In both cases, $(u,i)$ is counted in summand~(A1).

		\item If~$u \in F_i \cap F_{i-1}$, then~$u$ is dominated by some activity interval~$(w,a,b)$ with~$a \leq i-1$ and~$w \in \overleftarrow{F}_{i-1}$.
		Therefore~$\mathcal{T}'$ dominates~$u$ and therefore~$u \in \mathcal{A}'$.
		If~$u \in  F_{i} \setminus F_{i-1}$, then $\mathcal{A}\setminus \mathcal{A}'$ possibly contains temporal vertices $(u,j)$ for some~$j \in [i-1]$.
Since $u$ is isolated in these time steps, we conclude that~$(u,j) \in \mathcal{A}\setminus \mathcal{A}'$ for $j \in [i-1]$ if and only if $(u,a,b) \in \mathcal{T}$ with ${a \leq j \leq b}$.
If $\pos(u) = 0$, then there are at most $\curr(u)$ activity intervals of $u$ from $\mathcal{T}$ in~$G_1, \dots, G_{i-1}$ and therefore at most $\min(\curr(u)(\ell + 1), i-1)$ isolated temporal vertices of $u$ from $G_1, \dots, G_{i-1}$ are contained in $\mathcal{A} \setminus \mathcal{A}'$ (summand~(A2)).
If $\pos(u) > 0$, then there are at most $\curr(u) - 1$ activity intervals and a part of the $\curr(u)$-th activity interval of $u$ from $\mathcal{T}$ in $G_1, \dots, G_{i-1}$.
So in this case there are at most~$\min((\curr(u)-1)(\ell + 1) + \pos(u) - 1, i-1)$ isolated temporal vertices of $u$ from $G_1, \dots, G_{i-1}$ contained in $\mathcal{A} \setminus \mathcal{A}'$ (summand~(A3)).

		\item If~$u \in \overline{F}_i \cap \overline{F}_{i-1}$, then~$u$ is dominated by some activity interval of itself or by some~$(w,a,b)$ with~$a \leq i-1$ and~$w \in \overleftarrow{F}_{i-1}$.
		Again, by definition~$\mathcal{T}'$ dominates~$u$ and therefore~$u \in \mathcal{A}'$.

		If~$u \in \overline{F}_i \setminus \overline{F}_{i-1} = F_{i-1} \setminus F_i$, then $\mathcal{A}\setminus \mathcal{A}'$ possibly contains temporal vertices $(u,j)$ for $j \in [i,T]$.
Since $u$ is isolated in these time steps, we have $(u,j) \in \mathcal{A}\setminus \mathcal{A}'$ for $j \in [i, T]$ if and only if $(u,a,b) \in \mathcal{T}$ with $a \leq j \leq b$.
If $\pos(u) = 0$, then there are at most $(k - \curr'(u))$ activity intervals of $u$ from $\mathcal{T}$ in $G_i, \dots, G_T$ and therefore at most $\min((k-\curr'(u))(\ell + 1), T - i + 1)$ isolated temporal vertices of $u$ in $G_i, \dots, G_T$ are contained in $\mathcal{A} \setminus \mathcal{A}'$ (summand~(A4)).
If $\pos(u) > 0$, then there are at most~$(k - \curr'(u))$ activity intervals and a part of the $(\curr(u)-1)$-th activity interval of $u$ from $\mathcal{T}$ in $G_i, \dots, G_T$.
Consequently, in this case there are at most~$\min((k - \curr'(u))(\ell + 1) + \ell + 1 - \pos(u), T -i+1)$ isolated temporal vertices of $u$ in $G_i, \dots, G_T$ contained in $\mathcal{A} \setminus \mathcal{A}'$ (summand~(A5)).
	\end{enumerate}
	Combining the previous observations, we have $$\DP [i,\curr,\pos] = |\mathcal{A}| \leq |\mathcal{A}'| + (A1) + (A2) + (A3) + (A4) + (A5)$$ and therefore the desired inequality holds.

	($\geq$) Suppose the maximum on the right side of the computation is attained for~$\curr',\pos'$ and the maximum in the definition of $\DP[i-1, \curr', \pos']$ is attained for $\mathcal{T}'$.
We define $\mathcal{T}$ by extending $\mathcal{T}'$ in the following way.
	\begin{itemize}
		\item For $v \in F_{i-1} \cap F_i$ the activity interval $(v,i,\min(i+ \ell, T))$ is added if and only if $\pos (v) = 1$.
With these intervals $\mathcal{T}$ dominates all of $N_{G_i}[\{v \in F_i : \pos(v) \geq 1\}]$.
		\item For $v \in F_i \setminus F_{i-1}$ (so $v$ is isolated in the snapshots $G_1, \dots, G_{i-1}$) there are two cases.
If ${\pos(v) = 0}$, then $\curr(v)$ activity intervals of $v$ are added in the time interval $[i-1]$ such that $\min(\curr(v)(\ell + 1), i-1)$ isolated temporal vertices of~$v$ in $G_1, \dots, G_{i-1}$ are dominated by $\mathcal{T}$.
Otherwise, if $\pos(v) > 0$, then $(\curr(v)-1)$ activity intervals of $v$ are added in the time interval $[i-1]$ such that ${\min((\curr(v)-1)(\ell + 1) + \pos(v) -1, i-1)}$ isolated temporal vertices of $v$ in $G_1, \dots, G_{i-1}$ are dominated by $\mathcal{T}$.
		\item For $v \in \overline{F}_i \setminus \overline{F}_{i-1} = F_{i-1} \setminus F_i$ (so $v$ is isolated in $G_i, \dots, G_T$) there are also two cases.
If $\pos'(v) = 0$, then $(k -\curr'(v))$ activity intervals of $v$ are added in the time interval $[i, T]$ such that $\min((k-\curr'(v))(\ell + 1), T-i+1)$ isolated temporal vertices of $v$ in $G_i, \dots, G_T$ are dominated by $\mathcal{T}$.
If $\pos'(v) > 0$, then then $(k -\curr'(v))$ activity intervals of $v$ are added in the time interval $[i,T]$ so that $\min((k -\curr'(v))(\ell + 1) + \ell + 1 -\pos(v), i-1)$ isolated temporal vertices of~$v$ in $G_i, \dots, G_T$ are dominated by $\mathcal{T}$.
	\end{itemize}
	By definition of the table entries we have $$\DP[i, \curr, \pos] \geq \DP[i-1, \curr', \pos'] + (A1) + (A2) + (A3) + (A4) + (A5),$$ because $\mathcal{T}$ is considered in the definition of $\DP[i, \curr, \pos]$.
Therefore the desired inequality holds.

	\emph{Result:}
	Finally, we return yes if and only if $\DP[T, \curr, \pos] \geq t$ for $\curr : F_T \rightarrow \{k\}$ and some function $\pos : F_T \rightarrow \{0, \ell + 1\}$.
The restriction to these functions is correct, since we can assume without loss of generality that a vertex is active in $G_T$ if and only if its~$k$-th activity interval ends at time step $T$.
Further, note that $\overline{F}_T \cup F_T = V(\mathcal{G})$.
Consequently, all $k$-activity $\ell$-timelines of $\mathcal{G}$ are considered in the definition of the table entry $\DP[T, \curr, \pos]$ and all dominated temporal vertices of $\mathcal{G}$ are counted in the respective maximum of the definition.

	\emph{Running Time:}
	For each time step $i \in [T]$, there are $\big((k+1)(\ell + 2)\big)^{\vimw}$ choices for the functions $\curr : F_i \rightarrow [0,k]$ and $\pos : F_i \rightarrow [0, \ell + 1]$.
This upper-bounds the size of our dynamic programming table by $T \cdot \big((k+1)(\ell + 2)\big)^{\vimw}$.
In each recursive computation of a table entry, there are also $\Oh( \big((k+1)(\ell + 2)\big)^{\vimw})$ choices for $\curr'$ and $\pos'$.
The remaining summands of the recursive formula can be computed in~$\Oh(n)$ time.
Hence, the overall running time is bounded by~$\big((k+1)(\ell + 2)\big)^{2\cdot \vimw} \cdot (n+T)^{\Oh(1)}$.
\end{proof}

Recall that for \PCT we showed that~$\vimw=\vimw[1]$ can be replaced by the much smaller parameter~$\vimw[k\cdot(\ell+1)+1]$, the size of the $k\cdot(\ell+1)+1$~largest bag.
Hence, one may ask whether~$\vimw$ can be replaced by~$\vimw[i]$ with some suitably chosen~$i$ for \DT and \PDT.
We show that both problems behave quite differently:
First, we show we that~$\vimw$ can be replaced by~$\vimw[\ell+2]$ for \DT by having a preprocessing step and then adapting the algorithm behind \Cref{DomTimeLine_vimw+k+l}.
Second, we show that replacing~$\vimw=\vimw[1]$ by~$\vimw[2]$ is not possible for \PDT.
More precisely, we show NP-hardness even if there are only two snapshots one of which is edgeless (which implies~$\vimw[2]=0$).

We start with our result for \DT.

For \PDT, we show that we cannot replace~$\vimw=\vimw[1]$ by~$\vimw[2]$, since we prove para-NP-hardness for~$\vimw[2]+k+\ell$.
We start with our algorithmic result for \DT.

\begin{proposition}

	\label{DomTimeLine_vimw+k+l2}
	\DT is solvable in $\big((k+1)(\ell + 2)\big)^{4\cdot \vimw[x]} \cdot (n+T)^{\Oh(1)}$~time where~$x=\ell+2$.
\end{proposition}

\begin{proof}
Similarly to the proofs before, it is safe to assume that~$T\ge k\cdot(\ell+1)+1$ because otherwise each vertex can be active in each snapshot and thus every vertex is dominated.
Let~$(F_i)_{i\in [T]}$ be the bags of the vertex-interval-membership sequence of~$\mathcal{G}$.
Analogously to \Cref{CovTimeLine_vimw+k+l2}, we denote the $\ell+1$~largest bags of~$(F_i)_{i\in [T]}$ as \emph{large} and each remaining bag is denoted as \emph{small}.
The observation is that for each vertex~$v$ which is only contained in large bags, there are a huge number of snapshots were~$v$ is isolated.
To dominate~$v$ in such snapshots, $v$ has to be active during them.
This essentially fixes all $k$-activity intervals of~$v$.
Afterwards, we use a slight modification of the dynamic program of \Cref{DomTimeLine_vimw+k+l} to find the best activity intervals for the remaining vertices.

Next, let~$[i,j]\subseteq [T]$ be a maximal consecutive sequence of indices such that each bag~$F_x$ is large for any~$x\in[i,j]$.
Furthermore, let~$\mathcal{A}_{ij}=\bigcup_{x\in[i,j]}F_x\setminus(F_{i-1}\cup F_{j+1})$. The set~$\mathcal{A}_{ij}$ contains the vertices which are isolated in~$G_1, \ldots, G_{i-1}$ and in~$G_{j+1}, \ldots, G_T$.
Note that if~$i=1$ or~$j=T$ the respective sets~$F_{i-1}$ and~$F_{j+1}$ are empty.
Moreover, if~$\mathcal{A}_{ij}$ is empty, then we can conclude~$|F_x|\le |F_{i-1}|+|F_{j+1}|\le 2\cdot \vimw[\ell+2]$ for all~$x \in [i,j]$ and subsequently \Cref{DomTimeLine_vimw+k+l} yields the desired running time.
Hence, we assume that at least one such~$\mathcal{A}_{ij}$ is non-empty.

Observe that since~$v\in\mathcal{A}_{ij}$ is isolated in~$G_1, \ldots, G_{i-1}$ and in~$G_{j+1}, \ldots, G_T$, the vertex~$v$ must be active during~$[i-1]$ and during~$[j+1,T]$ to ensure that it is dominated at these time steps.
Since by definition~$j-i\le \ell$, in total these are at least $T-(\ell+1)\ge k\cdot(\ell+1)+1-(\ell+1)=(k-1)\cdot(\ell+1)+1$~snapshots.
Since $(k-1)$~activity intervals can dominate at most $(k-1)(\ell+1)$ of these isolated vertices, we conclude that all $k$-activity intervals of~$v$ are necessary  to dominate~$v$ in these snapshots.
In particular, if $T\ge k\cdot(\ell+1)+(\ell+1)+1$, then we deal with a no-instance since in at least one snapshot where~$v$ is isolated, no activity interval of~$v$ can be active.

Now, we can assign the $k$-activity intervals of~$v$ greedily:
To dominate~$v$ in~$G_1, \ldots, G_{i-1}$ we let the first activity interval of~$v$ start at snapshot~1, the second at snapshot~$\ell+2$ and so on.
By the observation above, we know that all~$k$ activity intervals of~$v$ are used in this way.
(If there is still some snapshot where~$v$ is isolated and not dominated, then return no.)

Recall that for \PCT in the proof of \Cref{CovTimeLine_vimw+k+l2} we could simply remove the vertices in~$\mathcal{A}_{ij}$ and then call the dynamic program of \Cref{CovTimeLine_vimw+k+l} for the smaller bag size.
But for \DT, however, we cannot simply remove the vertices in~$\mathcal{A}_{ij}$, since each vertex in~$\mathcal{A}_{ij}$ also needs to be dominated in snapshots~$G_i,\ldots , G_j$ and may also dominate some other vertices in~$G_i,\ldots, G_j$ if it is active.
\textbf{1)} Note that in some of these snapshots the vertices of~$\mathcal{A}_{ij}$ are active due to our greedy choice of the activity intervals for the vertices in~$\mathcal{A}_{ij}$; but since~$T\ge k\cdot(\ell+1)+1$, in at least one snapshot the vertices still need to be dominated.
\textbf{2)} Moreover, in each snapshot~$G_x$ where~$x\in[i,j]$ and the vertices in~$\mathcal{A}_{ij}$ are active, all temporal vertices in~$N_{G_{x}}[\mathcal{A}_{ij}]$ are already dominated.
Also note that in snapshot~$G_x$ for any~$x\in[i,j]$ either all vertices of~$\mathcal{A}_{ij}$ are active or none of them.

Due to Properties \textbf{1)} and \textbf{2)} we adapt the dynamic program of the algorithm of \Cref{DomTimeLine_vimw+k+l} as follows:
For each consecutive sequence~$F_i, \ldots, F_j$ of large bags we let~$\mathcal{A}_{ij}=\bigcup_{x\in[i,j]}F_x\setminus(F_{i-1}\cup F_{j+1})$ be the vertices which are isolated in~$G_1,\ldots, G_{i-1}$ and~$G_{j+1},\ldots, G_T$.
Also, let~$\mathcal{A}$ be the union of all sets~$\mathcal{A}_{ij}$, that is, the set of all vertices which are only non-isolated in some large bags.
In the new dynamic program we only consider the subset~$Q_i=F_i\setminus \mathcal{A}$ of vertices.
Consequently, for a small bag there is no change.
But for a large bag~$F_x$ we now only consider vertices~$Q_x=F_x\setminus \mathcal{A}$.
Observe that each vertex in~$Q_x$ is either contained in~$F_{i-1}$ or in~$F_{j+1}$.
Thus, $|Q_x|\le |F_{i-1}|+|F_{j+1}|\le 2\cdot \vimw[\ell+2]$.
We next describe how we change the update for bag~$F_x$ to respect \textbf{1)} and \textbf{2)}.
Recall that in~$G_x$ all vertices of~$\mathcal{A}\setminus\mathcal{A}_{ij}$ are isolated and thus active and that the vertices in~$\mathcal{A}_{ij}$ are either all active or not.
\textbf{1)} if no vertex in~$\mathcal{A}_{ij}$ is active, the active vertices in~$F_x\setminus \mathcal{A}_{ij}$ need to dominate~$\mathcal{A}_{ij}$.
Let~$S_x=\{v \in F_x: \pos(v) \geq 1\}$.
If~$\mathcal{A}_{ij}\subseteq N_{G_x}[S]$, then nothing changes, but otherwise if~$\mathcal{A}_{ij}\not\subseteq N_{G_x}[S]$, then we set the corresponding entry~$\DP[i, \curr, \pos]$ to~$-\infty$ since at least one vertex in~$\mathcal{A}_{ij}$ does not get dominated in~$G_x$.
\textbf{2)} if all vertices in~$\mathcal{A}_{ij}$ are active, then all vertices of~$N_{G_x}[\mathcal{A}_{i,j}]$ are already dominated, and thus we change summand~$(A1)$ to~$|N_{G_x} [\mathcal{A}_{ij}\cup \{v \in F_x: \pos(v) \geq 1\} ]|$.

The correctness and the running time of the updated dynamic program can be shown analogously as in the proof of \Cref{DomTimeLine_vimw+k+l}.
Only the exponent of the new running time is different since in each big bag we consider at most $2\cdot \vimw[\ell+2]$~vertices.
\end{proof}

Finally, we show hardness for constant~$k, \ell$, and~$\vimw[2]$ for \PDT.

\begin{theorem}

\label{PDomTimeLine_vimw2+k+l}
\PDT is NP-hard even if~$T=2$, $k=1$, $\ell=0$, the underlying graph is bipartite, planar, has maximum degree~3, and one snapshot is edgeless.
\end{theorem}

\begin{proof}
We reduce from the classic \DS problem, where the input consists of a graph~$H$ and an integer~$k$ and the goal is to find a dominating set of size most~$k$.
Here, a vertex set~$S$ is a dominating set if~$N[S]=V(H)$.
\DS is NP-hard even if the graph is bipartite, planar and has a maximum degree of~3~\cite{GareyJohnsonBook}.

Now, we construct a \PDT instance~$(\mathcal{G},k=1,\ell=0,t)$ as follows:
We set~$T=2$.
The first snapshot~$G_1$ of~$\mathcal{G}$ is a copy of~$H$ and the second snapshot~$G_2$ is edgeless.
Clearly, the underlying graph of~$\mathcal{G}$ has maximum degree~3.
Finally, we set~$t=2n-k$ where~$n=|V(H)|$.

If~$S$ is a dominating set, then letting~$S$ be active in~$G_1$, and letting the vertices~$V(H)\setminus S$ be active in~$G_2$ dominates $t=2n-k$~vertices of~$\mathcal{G}$.
Conversely, suppose we have a solution for~$(\mathcal{G},k=1,\ell=0,t=2n-k)$.
Let~$A$ be the vertices which are active in~$G_1$ and let~$B$ be the vertices which are active in~$G_2$.
Clearly, $|A|+|B|\leq n$.
Let~$\alpha=|N[A]|$ and~$\beta=|B|$.
By definition~$\alpha+\beta\ge t = 2n-k$.
Since~$\alpha\le n$ and~$\beta\le n$, we obtain~$\alpha\ge n-k$, $\beta\ge n-k$ and thus~$|B|\ge n-k$ and~$|A|\le k$.
Hence, if~$\alpha=n$, $A$ is a dominating set of size at most~$k$ and we are done.
Consequently, in the following it is safe to assume that~$\alpha<n$.
We let~$Z=V(H)\setminus N[A]$ be the vertices of~$G_1$ which are not dominated by~$A$.
Moreover, we let~$\gamma=\beta-(n-k)$.
Note that~$0\le\gamma\le k$.
Now, observe that if~$|Z|>\gamma$, $A$ and~$B$ dominate only $n-|Z|+\gamma +n-k\le 2n-k-1$~vertices, a contradiction.
Thus, $|Z|\le \gamma$.
Now, we claim that~$S=A\cup Z$ is a dominating set of~$H$ of size at most~$k$.
By definition, $S$ is a dominating set.
For the size bound, observe that~$k =\gamma+n-\beta$ and since~$\gamma\ge 0$ we obtain that~$S$ has size at most~$k$.
\end{proof}

\subsection{The Parameter Interval-Membership-Width}

Now, we study the influence of~$\imw$ instead of~$\vimw$.
Initially, we show that the parameter~$\imw+\,k+\ell$ yields a polynomial kernel for \DT.
This also implies a polynomial kernel for~$\vimw+\,k+\ell$ and~$n+k+\ell$ for \DT.
More precisely, we show an even stronger result:
We provide a polynomial kernel for~$q+k+\ell$, where $q$ is the maximal number of edges in any snapshot.
Note that~$q$ is never larger than~$\imw$, since any bag~$F_i$ contains~$E_i$. Moreover, in the temporal graph~$\mathcal{G}$ with a star on~$n$ leaves as underlying graph and lifetime $T=2n$, where the edge towards the~$i$-th leaf is active at snapshots~$i$ and~$i+n$, we have~$\imw=n$ and~$q=1$.

\begin{theorem}
\label{thm-kernel-imw-k-l}
\DT has a kernel where both the number of snapshots and vertices are cubic in~$q+k+\ell$ which can be computed in linear time with respect to the input size.
\end{theorem}
\begin{proof}
In order to provide the kernel we bound the number of snapshots~$T$ and the number~$n$ of vertices of the underlying graph based on the following case distinction.

\emph{Case~1: $q\ge n/2$.} 
By definition each $k$-activity $\ell$-timeline has active vertices in at most $n\cdot k\cdot (\ell+1)$~snapshots.
Hence, if~$T> 2\cdot q\cdot k\cdot (\ell+1) \ge n\cdot k\cdot (\ell+1)$, then the instance is a no-instance. 

\emph{Case~2.1: $q < n/2$ and $T> 2\cdot k\cdot (\ell+1)+1$.} 
We show that in this case, we are dealing with a no-instance. Since each snapshot contains at most $q$~edges, any set of~$p$ vertices can dominate at most~$p + q$~vertices in this snapshot.
As~$q < n / 2$, it follows that at least~$n/2$~vertices must be active in order to dominate all temporal vertices of the corresponding snapshot~$G_x$.
Moreover, note that the total number of active vertices in any $k$-activity $\ell$-timeline is at most~$z=n\cdot k\cdot (\ell+1)$.
Consequently,~$T\le z/(n/2)\le 2\cdot k\cdot (\ell+1)$ for any yes-instance.

\emph{Case~2.2: $q < n/2$ and $T\le 2\cdot k\cdot (\ell+1)+1$.}
If~$n\le 4\cdot q\cdot k\cdot (\ell+1)$, then we have the desired bound. 
Otherwise, if~$n \ge 4\cdot q\cdot k\cdot (\ell+1) +1$, we can solve the instance in polynomial time as follows. Since each snapshot has at most $q$~edges, at most $2q$~vertices are non-isolated in each snapshot.
Moreover, since there are at most $2\cdot k\cdot (\ell+1)$~snapshots, in total at most $4\cdot q\cdot k\cdot (\ell+1)$~vertices are non-isolated in the temporal graph.
Consequently, since~$n \ge 4\cdot q\cdot k\cdot (\ell+1) +1$ there exists at least one vertex~$v$ which is isolated in every snapshot.
Thus, if~$T\ge k\cdot(\ell+1)+1$ we deal with a trivial no-instance and if~$T\le k\cdot(\ell+1)$ we deal with a trivial yes-instance.

Thus, in each case we either solve the instance (giving a kernel of constant size) or obtain a kernel with at most $2\cdot q\cdot k\cdot (\ell+1)$~snapshots and at most ${4\cdot q\cdot k\cdot(\ell+1)}$~vertices.
Moreover, the kernel can be computed in linear time as it involves only computing~$n$,~$q$, and~$T$.
\end{proof}

We show that the remaining problems are NP-hard if~$\imw + \,k+\ell \in \Oh(1)$.

\begin{theorem}

\label{thm-cov-np-h-imw-k-l}
\CT is NP-hard even if~$k=2$, $\ell=0$, $\imw = 4$, each edge of the underlying graph exists in exactly one snapshot, and the underlying graph has a constant maximum degree.
\end{theorem}

\begin{proof}
We reduce from the NP-hard \prb{$3$-Coloring} on graphs with maximum degree four~\cite{DBLP:conf/stoc/GareyJS74}.

	\textit{Construction}: Let $(G= (V,E))$ be an instance of \prb{$3$-Coloring} on graphs with maximum degree four.
We construct the following temporal graph~$\mathcal{G}$.
	For each vertex $v \in V$, we fix an arbitrary ordering of the edges incident with~$v$ and add $16$ vertices.
Formally, we set
	\begin{align*}
		V(\mathcal{G}) := \bigcup_{v\in V}  \, \{r_v^i,y_v^i,b_v^i, x_v^i\colon i\in [4]\}.
	\end{align*}
	The vertex sets~$R_v := \{r_v^i\colon i\in [4]\}$, $Y_v := \{y_v^i\colon i\in [4]\}$, and~$B_v := \{b_v^i\colon i\in [4]\}$ will indicate the color choices of the vertex.
	The temporal graph $\mathcal{G}$ consists of $12|V| + 3|E|$ snapshots.
In the following we formally define the snapshots.
Note that the order of the snapshots does not matter for the instance, since~$\ell = 0$ and each edge will appear only in a single snapshot.

	For each vertex~$v\in V$, each color~$c\in \{r,y,b\}$, and each~$i\in [4]$, we add a snapshot~$G_{v,c,i}$ with edge set~$\{x_v^ic_v^j\colon j\in [4]\}$.
 For each edge~$e=\{u,v\}\in E$ and each color~$c\in \{r,y,b\}$, we add a snapshot~$G_{e,c}$ that has only the single edge~$c_u^ic_v^j$, such that~$e$ is the~$i$th edge incident with~$u$ in~$G$ and the~$j$th edge incident with~$v$ in~$G$.
	This completes the construction of~$\mg$.
Note that each edge of the underlying graph exists in exactly one snapshot:
Edges incident with a vertex~$x_v^i$ only exist in the three snapshots~$G_{v,r,i}$, $G_{v,y,i}$, and~$G_{v,b,i}$, and all these three snapshots are pairwise edge disjoint.
All other edges are part of only a single snapshot by construction of the snapshots~$G_{e,c}$.
Moreover, note that each vertex~$q$ of~$R_v\cup Y_v\cup B_v$ has at most one snapshot outside of~$\{G_{v,c,i}\colon c\in \{r,y,b\}, i \in [4]\}$ in which it has an active incident edge.
If such a snapshot exists, we may call it the~\emph{relevant} snapshot of~$q$.
Finally, we set~$k=2$ and~$\ell=0$.

\textit{Intuition}:
Let~$v\in V$.
For each~$i\in [4]$, to cover all edges of the snapshots~$\{G_{v,c,i}\colon c\in \{r,y,b\}\}$ the temporal vertex~$x_v^i$ can be used only twice, which implies that for some~$C\in \{R,Y,B\}$, all vertices of~$C_v$ are part of the timeline~$\mt$ which covers all temporal edges of~$\mg$ for at least one of the snapshots~$\{G_{v,c,i}\colon c\in \{r,y,b\}\}$.
Since~$i$ has 4 possible values and each vertex can be part of~$\mt$ for at most two snapshots, there is some~$C\in \{R,Y,B\}$, such that all vertices of~$C_v$ in~$\mt$ are from the snapshots~$\{G_{v,c,i}\colon c\in \{r,y,b\},i\in[4]\}$.
As we will show, this resembles a color choice for which we can show that it implies a proper~$3$-coloring of~$G$.
Intuitively, this is the case since for the selected color~$c$, for each edge~$e=vw$ of~$G$, the unique temporal edge of~$G_{e,c}$ must be covered, which can only be achieved by a vertex of~$C_w$.
Thus, the color of~$w$ cannot be~$c$ simultaneously.

\textit{Correctness}:
We now show that~$G$ is 3-colorable if and only if there is a~$k$-activity $\ell$-timeline covering all temporal edges of~$\mg$.

$(\Rightarrow)$
Let~$\pi\colon V \to \{r,y,b\}$ be a proper~$3$-coloring of~$G$.
We define a $k$-activity $\ell$-timeline~$\mt$ covering all temporal edges of~$\mathcal{G}$ as follows.
We initialize~$\mt=\emptyset$.
For a fixed vertex~$v\in V$ we let~$c := \pi(v)$, and we let~$\{\alpha,\beta\} = \{r,y,b\} \setminus \{c\}$.
We add the following temporal vertices to~$\mt$:
\begin{enumerate}
\item\label{item 1} the vertices of~$\{c_v^i\colon i\in [4]\}$ from both snapshots~$G_{v,c,1}$ and~$G_{v,c,2}$,
\item\label{item 2} the vertex~$x_v^1$ from both snapshots~$G_{v,\alpha,1}$ and~$G_{v,\beta,1}$,
\item\label{item 3} the vertex~$x_v^2$ from both snapshots~$G_{v,\alpha,2}$ and~$G_{v,\beta,2}$,
\item\label{item 4} the vertex~$x_v^3$ from both snapshots~$G_{v,c,3}$ and~$G_{v,\alpha,3}$,
\item\label{item 5} the vertex~$x_v^4$ from both snapshots~$G_{v,c,4}$ and~$G_{v,\beta,4}$,
\item\label{item 6} the vertices of~$\{\alpha_v^i\colon i\in [4]\}$ from snapshot~$G_{v,\alpha,4}$,
\item\label{item 7} the vertices of~$\{\beta_v^i\colon i\in [4]\}$ from snapshot~$G_{v,\beta,3}$, and
\item\label{item relevant} each vertex~$q$ of~$\{\alpha_v^i,\beta_v^i\colon i\in [4]\}$ from its relevant snapshots (if such a snapshot exists).
\end{enumerate}
Note that each vertex of the underlying graph is part of at most two snapshots in~$\mt$.
We now show that in each snapshot, for each temporal edge, at least one endpoint is in~$\mt$.

To this end, observe that for each vertex~$v$, each temporal edge of any of the snapshots of~$\{G_{v,c,i}\colon c\in \{r,y,b\}, i \in [4]\}$ is covered by exactly one temporal vertex described in \Cref{item 1,item 2,item 3,item 4,item 5,item 6,item 7}.
Thus, consider a temporal edge~$(e',t)$ that is not incident with any vertex~$x_w^i$.
By definition of all snapshots outside of~$\{G_{v,c,i}\colon c\in \{r,y,b\}, i \in [4]\}$, $e' = c_u^ic_v^j$ for some edge~$e=uv\in E$ and a color~$c\in \{r,y,b\}$, such that~$e$ is the~$i$th edge incident with~$u$ in~$G$ and the~$j$th edge incident with~$v$ in~$G$.
In particular, $G_{e,c}$ is the unique snapshot in which~$e'$ is active, implying that this is the relevant snapshot of both~$c_u^i$ and~$c_v^j$.
Since~$\pi$ is a proper coloring of~$G$, there at least one~$w\in \{u,v\}$, such that~$\pi(w) \neq c$.
Without loss of generality assume that~$w=v$.
Thus, by~\Cref{item relevant}, $(c_v^j,t)\in \mt$.
Hence, the temporal edge~$(e',t)$ is covered, implying that~$\mt$ is a~$k$-activity $\ell$-timeline covering all temporal edges of~$\mathcal{G}$.

$(\Leftarrow)$
Let~$\mt$ be a $k$-activity $\ell$-timeline that covers all temporal edges of~$\mg$.
First, we show that for each vertex~$v\in V$, there is at least one color~$C \in \{R,Y,B\}$, such that none of the vertices of~$C_v$ is part of~$\mt$ in its respective relevant snapshot.
To this end, note that for each~$i\in [4]$, $x_v^i$ can only be part of~$\mt$ for at most two of the three snapshots~$G_{v,r,i}$, $G_{v,y,i}$, and~$G_{v,b,i}$.
This is due to the fact that~$k = 2$.
Since in each of these three snapshots, the edge set corresponds to a star with center~$x_v^i$ and leaf set~$\{c_v^i\colon i\in [4]\}$ for each~$c\in \{r,y,b\}$, this implies that, in order to have a vertex cover in each snapshot, for at least one set~$C_v\in \{R_v,Y_v,B_v\}$, all vertices of~$C_v$ are part of~$\mt$ for at least one snapshot of~$G_{v,r,i}$, $G_{v,y,i}$, and~$G_{v,b,i}$.
Since this property is true for each~$i\in[4]$, for at least one of the sets~$C_v\in \{R_v,Y_v,B_v\}$, all vertices of~$C_v$ use both possible appearances in~$\mt$ in the snapshots of~$\{G_{v,c,i}\colon c\in \{r,y,b\}, i \in [4]\}$.

With this property at hand, we can now define a~$3$-coloring of~$G$.
Let~$\pi\colon V \to \{r,y,b\}$ where for each vertex~$v\in V$, $\pi(v)$ is an arbitrary fixed color~$c\in \{r,y,b\}$ for which all vertices of~$\{c_v^j\colon j\in [4]\}$ use both possible appearances in~$\mt$ in the snapshots of~$\{G_{v,c,i}\colon c\in \{r,y,b\}, i \in [4]\}$.
By the above argumentation, $\pi$ is properly defined.
It remains to show that for each edge~$uv\in E$, $\pi(u) \neq \pi(v)$.
Assume towards a contradiction that there is an edge~$e=uv\in E$ with~$\pi(u) = \pi(v)$.
Let~$c := \pi(u)$, and let~$i$ and~$j$ be such that~$e$ is the~$i$th edge incident with~$u$ in~$G$ and~$e$ is the~$j$th edge incident with~$v$ in~$G$.
Then, by definition, $G_{e,c}$ is a snapshot of~$\mg$ that contains exclusively the edge~$c_u^ic_v^j$.
By definition of~$\pi$ and since~$\pi(u) = \pi(v) = c$, for each of the two vertices~$c_u^i$ and~$c_v^j$, both possible appearances in~$\mt$ are only in the snapshots of~$\{G_{v,c,i}\colon c\in \{r,y,b\}, i \in [4], v\in V\}$.
Hence, $c_u^ic_v^j$ is not covered by~$\mt$.
This contradicts the fact that~$\mt$ covers all temporal edges of~$\mg$.
Thus, we obtain~$\pi(u) \neq \pi(v)$ which implies that~$\pi$ is a proper~$3$-coloring of~$G$.

\end{proof}

\begin{theorem}

\label{thm-pdt-np-h-imw-k-l}
	\PDT is NP-hard even if~$k=1$, $\ell=0$, $\imw = 36$ and the underlying graph has a constant maximum degree.
\end{theorem}

\begin{proof}
	We give a reduction from \prb{3-Sat-(2,2)}.
	In this 3-SAT variant each variable appears exactly two times positively and exactly two times negatively.
	This problem is NP-hard \cite{DBLP:journals/dam/DarmannD21}.

	\textit{Construction}: Let~$\phi$ be an instance of \prb{$3$-Sat-(2,2)} with variable set~$X := \{x_1, \dots, x_n\}$ and clause set~$C:= \{C_1, \dots, C_m\}$.
	We fix some order on the set of clauses which induces an order on the (non-)negated appearances of each variable.
	As a result, we can denote the first and second nonnegated appearance of~$x_i$ by~$x_{i}^1$ and~$x_{i}^2$ and the first and second negated appearances by~$\overline{x}_{i,1}$ and~$\overline{x}_{i,2}$.
	We construct the following temporal graph~$\mg$.
	For each variable~$x_i$ we introduce two sets with seven vertices
	\begin{equation*}
		V_i^1 := \{x_i^1, \overline{x}_i^1, p_i^1, q_i^1, r_i^1, s_i^1, t_i^1\}, \\
		V_i^2 := \{x_i^2, \overline{x}_i^2, p_i^2, q_i^2, r_i^2, s_i^2, t_i^2\}
	\end{equation*}
	and add six snapshots~$G_{6i - 5}, \dots, G_{6i}$.
	The first three of these six snapshots only contain a crown graph on~$(V_i^1 \cup V_i^2) \setminus \{\overline{x}_i^1, \overline{x}_i^2\}$, that is, a biclique with partite sets~$V_i^1\setminus\{\overline{x}_i^1\}$ and~$V_i^2\setminus\{\overline{x}_i^2\}$ without the edges of the form~$v_i^1v_i^2$ where~$v\in\{x,p,q,r,s,t\}$.
	We remark that each pair of the form~${v^1v^2}$ in the crown graph dominates the whole crown graph.
	For the remaining three snapshots we do the same but for the set~$(V_i^1 \cup V_i^2) \setminus \{x_i^1, x_i^2\}$.
	In the following, these 6~snapshots for variable~$x_i$ are denoted by~$X_i$.
	We denote these as the \emph{variable snapshots}.
		Note that this constructs the first~$6n$ snapshots of our temporal graph~$\mg$.

	For each clause~$C_j$ we add six vertices
	\begin{equation*}
		U_j := \{a_j,a_j', b_j, b_j', c_j, c_j'\}
	\end{equation*}
	and introduce one snapshot~$G_{6n + j}$ to which we refer by~$G_{C_j}$.
	This snapshot contains a triangle on the three vertices corresponding to the literal appearances in~$C_j$ (recall that the vertex~$x_i^j$ stands for the $j$-th appearance of~$x_i$).
	Furthermore, we make each of the three vertices adjacent to two different vertices of~$U_j$.
	This constructs the snapshots~$G_{6n + 1}, \dots, G_{6n + m}$ which we call \emph{clause snapshots}.

	Finally, for every~$j \in [m]$ we add six more snapshots~$G_{6n +m + 6j - 5}, \dots, G_{6n + m + 6j}$.
	Here, the first three of the six snapshots only contain the triangle on~$\{a_j, b_j, c_j\}$ and the remaining three only contain the triangle on~$\{a_j', b_j', c_j'\}$.
	We refer to these six snapshots by~$Z_j$.
	These are the \emph{dummy snapshots}.
	This finishes the construction of our temporal graph~$\mg = (G_1, \dots, G_{6n + 7m})$.

	Observe that the interval-membership-width of~$\mg$ is exactly the number of edges in the crown graphs and therefore we have~$\imw = 36$.

	Finally, we set~$k = 1, \ell = 0, t = 76n + 21m$ and obtain the \PDT-instance~$\mathcal{I} = (\mg, k, \ell, t)$.

	\textit{Intuition}:
	Intuitively, the six variable snapshots for each variable form a selection gadget.
	In particular, the size of the crown graphs and the existence of the other vertices from~$V_i^1 \cup V_i^2$ force that either both~$x_i^1$ and~$x_i^2$ or both~$\overline{x}_i^1$ and~$\overline{x}_i^2$ are active in~$X_i$ and therefore not in the clause snapshots.
	Dominating the triangle in the snapshot~$G_{C_j}$ corresponds to satisfying the clause~$C_j$ and the vertices from~$U_j$ ensure that it is also worth to set two or three literals to true in one clause.
	The last~$6m$ dummy snapshots are constructed with the purpose to take care of the vertices from~$U_j$ such that they do not influence the clause snapshots.

	\textit{Correctness}:
	We proceed by showing that~$\phi$ is satisfiable if and only if there is a~$k$-activity $\ell$-timeline~$\mt$ for~$\mg$ which dominates at least~$t = 76n + 21m$ temporal vertices.

	$(\Rightarrow)$
	Let~$\alpha$ be a satisfying assignment of~$\phi$.
	If~$\alpha (x_i) = 1$, then we place the activity intervals of~$(V_i^1 \cup V_i^2) \setminus \{x_i^1, x_i^2\}$ such that all constructed crown graphs in~$X_i$ are dominated.
	Note that this is possible since~$|V _i^1 \cup V_i^2| = 14$ and there are six snapshots where one pair of vertices~$\{v_i^1, v_i^2\}$ for~$v\in \{\overline{x}, p, q, r, s ,t\}$ in each snapshot is enough to dominate the crown graphs.
	Moreover, we let~$x_i^1$ and~$x_i^2$ be active in the snapshot corresponding to the clause in which they are contained.
	This means if~$x_{i,1} \in C_j$, then we let~$x_i^1$ be active in the snapshot~$G_{C_j}$.
	If~$\alpha (x_i) = 0$, then we do the same but for~$V_i^1 \cup V_i^2 \setminus \{\overline{x}_i^1, \overline{x}_i^2\}$ and~$\overline{x}_i^1, \overline{x}_i^2$.
	In particular, this already dominates 76 temporal vertices for each variable via the crown graphs in all variable snapshots and the four temporal vertices from the sets~$U_j$ which are dominated in the clause snapshots.

	Since~$\alpha$ is a satisfying assignment, we conclude that in all clause snapshots~$G_{C_j}$ the triangle, which corresponds to the clause, is dominated.
	This gives us~$3m$ additionally dominated temporal vertices in the clause snapshots.

	Finally, for every~$j \in [m]$ we use the activities of the six vertices from~$U_j$ to dominate the triangle in every of the six snapshots in~$Z_j$.
	This yields 18 dominated temporal vertices for each~$j \in [m]$ in the dummy snapshots and therefore the solution dominates $ 76n + 21m$~temporal vertices in total.

	$(\Leftarrow)$
	Let~$\mathcal{T}$ be a solution to our \PDT-instance and suppose~$\mathcal{T}$ dominates at least $ 76n + 21m$~temporal vertices.

	First of all, observe that for each~$j \in [m]$ we can assume that every vertex from~$U_j$ is active in some dummy snapshot in~$Z_j$.
	This is due to the fact that every vertex from~$U_j$ has degree at most two in all snapshots and in~$Z_j$ each of them can dominate three temporal vertices without affecting any vertex outside of~$U_j$.
	This yields~$18m$ dominated temporal vertices.
	Consequently, no vertex of any~$U_j$ is active during any variable snapshot or any clause gadget.

	Second, observe that it is safe to assume that all vertices from~$V_i^1 \cup V_i^2$ are either active during a variable snapshot in~$X_i$ or in some clause snapshot since in every other snapshot each vertex in~$V_i^1 \cup V_i^2$ is isolated.
	We now show that in every snapshot in~$X_i$ exactly one pair~$\{v_i^1,v_i^2\} \subseteq V_i^1 \cup V_i^2$ for some~$v\in \{x,\overline{x}, p, q, r, s ,t\}$ is active.

	If there is a snapshot~$G_{j}$ associated with~$X_i$ in which no vertex of~$V_i^1 \cup V_i^2$ is active, then there exists a vertex~$v_i^1 \in V_i^1$ from the crown graph which is active $a)$ in a snapshot outside of~$X_i$ or $b)$ in a snapshot in~$X_i$ together with some other vertex of~$V_i^1$ since exactly 6 snapshots are associated with~$X_i$.
	In this case, making~$v_i^1$ active in~$G_{j}$ yields a solution which dominates strictly more vertices since in case $a)$ only one temporal vertex is dominated and in case $b)$ making~$v_i^1$ inactive reduced the number of dominated temporal vertices by only one.
	Thus, in the following, it is safe to assume that in each snapshot~$G_{j}$ associated with~$X_i$ at least one vertex of~$V_i^1$ is active.

Now, suppose there exist a snapshot~$G_j$ associated with~$X_i$ such that only vertices of~$V_i^1$ are active and no vertex of~$V_i^2$.
Since in each snapshot associated with~$X_i$ at least one vertex of~$V_i^1$ is active, in~$G_j$ at most 2~vertices of~$V_i^1$ are active, without loss of generality say~$p_i^1$ and~$q_i^1$.
Now, consider vertex~$p_i^2$.
Making~$p_i^2$ active during~$G_j$ dominates the 5~temporal vertices~$V_i^1\setminus\{p_i^1,q_i^1\}$.
Next, consider the snapshot~$G_z$ where~$p_i^2$ is currently active.
If~$G_z$ is not associated with~$X_i$, then~$p_i^2$ is isolated in~$G_z$.
Otherwise, if~$G_z$ is associated with~$X_i$ at least one vertex of~$V_i^1\setminus\{p_i^1,q_i^1\}$, say~$s_i^1$, is active during~$G_z$.
Hence, making~$p_i^2$ inactive at~$G_z$ only makes (a subset of the) 5~temporal vertices~$V_i^1\setminus\{p_i^1,s_i^1\}$ undominated.
Consequently, it is safe to assume that~$p_i^2$ is active during~$G_j$.
Thus, we have now shown that it is safe to assume that in every snapshot in~$X_i$ at least one pair~$\{v_i^1,v_i^2\} \subseteq V_i^1 \cup V_i^2$ for some~$v\in \{x,\overline{x}, p, q, r, s ,t\}$ is active.

Since any pair~$\{v_i^1,v_i^2\} \subseteq V_i^1 \cup V_i^2$ for some~$v\in \{x,\overline{x}, p, q, r, s ,t\}$ dominates all temporal vertices of the crown, it is safe to assume that per associated snapshot with~$X_i$ exactly one pair~$\{v_i^1,v_i^2\} \subseteq V_i^1 \cup V_i^2$ is active.
					As this holds for every~$i \in [n]$, we know that by construction at most~$4n + 3m$ temporal vertices in the clause snapshots can be dominated.
	Consequently, the solution~$\mathcal{T}$ has to dominate in total $76n$ temporal vertices in the parts~$X_i$.
	However, this can only be the case if exactly one pair~$\{v^1,v^2\}\subseteq V_i^1 \cup V_i^2$ is active in every snapshot of~$X_i$.
	Since each vertex from~$V_i^1 \cup V_i^2 \setminus \{x_i^1, x_i^2, \overline{x}_i^1, \overline{x}_i^2\}$ appears only isolated outside of~$X_i$ and each vertex from~$\{x_i^1, x_i^2, \overline{x}_i^1, \overline{x}_i^2\}$ dominates at least two vertices from~$U_j$ in the corresponding clause snapshot, it follows that either~$\{x_i^1, x_i^2\}$ or~$\{\overline{x}_i^1, \overline{x}_i^2\}$ are not active in~$X_i$ since~$\mathcal{T}$ maximizes the number of dominated temporal vertices.

	Finally, this allows us to define an assignment~$\alpha$ by setting~$\alpha (x_i) = 1$ if and only if~$\{x_i^1, x_i^2\}$ are not active in~$X_i$.
	To see that this is a satisfying assignment, recall that~$\mathcal{T}$ dominates at least $76n + 21m$ temporal vertices.
	By the reasoning above it follows that~$\mathcal{T}$ dominates in total at least~$4n + 3m$ temporal vertices in the clause snapshots.
	Clearly, this can only be the case if in every clause snapshot one of the vertices of the contained triangle is active, which by definition of~$\alpha$ means that~$\alpha$ sets at least one literal appearance in every clause to~$1$.
\end{proof}

\section{Complexity with respect to Input Parameters}
\label{sec-input-paras}

Finally, we study the influence of the input parameters~$n,k,\ell$, and~$t$ on the complexity of \PCTT and \PDTT.
Both \CT and \DT are NP-hard for constant values of~$k$ and~$\ell$.
Consequently, larger parameters need to be considered to obtain FPT-algorithms or XP-algorithms.
Froese et al.~\cite[Theorem~8]{DBLP:journals/mst/FroeseKZ24} showed that \CT admits an FPT-algorithm for~$n+k$.
A~similar algorithm also works for \DT.

\begin{theorem}

	\label{DomTimeline_n+k}
	\DT is solvable in time $\Oh(n^{2 + nk}T)$.
\end{theorem}

\begin{proof}
	Let $(\mathcal{G} = (G_1, \dots, G_T), k, \ell)$ be an instance of \DT.

	The algorithm starts with an empty solution candidate $\mathcal{T} = \emptyset$ and recursively adds activity intervals to this set.
Furthermore, it keeps one counter $k_v$ for each vertex $v \in V(\mathcal{G})$ for the number of activity intervals.

	If $k_v > k$ for some $v \in V$, then return no.
If $\mathcal{T}$ dominates all temporal vertices in $\mathcal{G}$, then return yes.
Otherwise, let $i$ be the minimum time step such that $G_i$ contains an undominated vertex $v$.
For each $u \in N_{G_i}[v]$ the algorithm branches by adding $(u,i,i+\ell)$ to~$\mathcal{T}$, increasing $k_u$ by one and continuing recursively.

	The size of the closed neighborhood of $v$ in $G_i$ is at most $n$.
In each recursive call the algorithm selects one activity interval and a solution can contain at most~$nk$ activity intervals.
Consequently, the search tree is of size $\Oh(n^{nk})$.
For each recursive call, the algorithm needs to find the first undominated temporal vertex.
This can be done in $\Oh(n^2T)$ time and therefore we obtain an overall running time of~${\Oh(n^{nk} \cdot n^2T)}$.
\end{proof}

Froese et al.~\cite[Theorem~12]{DBLP:journals/mst/FroeseKZ24} showed that \CT is W[1]-hard for~$n$ even if~$\ell=1$ with a reduction from \prb{Unary Bin Packing} and using a multicolored and a nonuniform variant of \CT as intermediate problems in the reduction.
In \prb{Nonuniform \CT} we are not given a single number of permitted activity intervals $k$, but instead a number $k_v$ for each vertex $v$.
In a similar way we introduce \prb{Nonuniform \DT}, which we will use as an intermediate problem to show hardness for~$n+\ell$ for \DT.

\begin{theorem}

	\label{W1hardnessN}
	\DT parameterized by $n$ is $\W[1]$-hard even if~$\ell = 1$.
\end{theorem}

We start by proving the hardness of the nonuniform variant.

\begin{lemma}
	\label{LemmaVCtoDStimeline}
	\prb{Nonuniform \DT} is $\W[1]$-hard with respect to $n$ even if $\ell = 1$.
\end{lemma}
\begin{proof}
	In the $\W[1]$-hardness proof of \CT, Froese et.
al \cite{DBLP:journals/mst/FroeseKZ24} use \prb{Nonuniform \CT} as an intermediate problem.
In particular, they show that this problem, when parameterized by $n$, is $\W[1]$-hard even if $\ell = 1$.
We now reduce from this problem.
A sketch of the idea can be found in \Cref{FigureLemmaVCtoDStimeline}.

	A common reduction from \prb{Vertex Cover} to \prb{Dominating Set} is to introduce a new vertex for each edge and make this new vertex adjacent to both endpoints of the edge.
In the following reduction we will make use of this idea.
However, we also need to take care of isolated temporal vertices.
For this purpose, we introduce another vertex which is connected to all isolated vertices in each snapshot.

	Suppose $(\mathcal{G} = (G_1, \dots, G_T), (k_v)_{v \in V}, \ell = 1)$ is a given instance of \prb{Nonuniform \CT} and let $V := V(\mathcal{G})$.
We construct a new temporal graph $\mathcal{G}'$ with vertex set $$V' := V \cup \{w_e :  e \in E_\downarrow (\mathcal{G})\} \cup \{z\}.$$ For each snapshot $G_i$ (with edge set $E_i$) of $\mathcal{G}$ we introduce the snapshot $G_i'$ for $\mathcal{G}'$ with the edge set
	\begin{align*}
		E_i' := E_i & \cup \{vw_e : e \in E_i, v \in e\} \\
		& \cup \{vz : v \in V, \deg_{G_i} (v) = 0\} \\
		& \cup \{w_ez : e  \in E_\downarrow (\mathcal{G}) \setminus E_i\}.
	\end{align*}
	In other words, the edge set of $G_i'$ consists of the edges $E_i$, the edges between the new vertices, which were introduced for the edges from $E_i$ and the respective endpoints and the edges, which connect all remaining isolated vertices to $z$.
	It remains to fix the individual number of activity intervals for each vertex.
For each $v \in V'$ we set
	\begin{align*}
		k'_{v} :=
		\begin{cases}
			k_v & \text{ if } v \in V, \\
			\lceil \frac{T}{2} \rceil & \text{ if } v = z, \\
			0 & \text{ otherwise.}
		\end{cases}
	\end{align*}
	The idea is that $z$ is active in each snapshot and dominates all isolated temporal vertices, while the vertices from $V$ dominate all the remaining temporal vertices.

	By the previous construction, the \prb{Nonuniform \DT} instance is given by $(\mathcal{G}' = (G_1', \dots, G_T'), (k'_v)_{v \in V'}, \ell = 1)$.
Since $|V'| \in O(|V|^2)$, the $\W[1]$-hardness follows if the reduction is correct.
	\begin{figure}[t]
		\centering
		\begin{tikzpicture}[scale = 0.6, transform shape,
			V/.style = {circle, draw, fill=black, inner sep = 6pt},
			d/.style = {circle, draw, fill=black, inner sep = 0.75pt},
			d2/.style = {circle, draw, fill=black, inner sep = 1.5pt}
			]

						\node[V, label=left:\large $v_1$, fill=red] (v1) at (-13, 2) {};
			\node[V, label=left:\large $v_2$] (v2) at (-13, 6) {};
			\node[V, label=right:\large $v_3$, fill=red] (v3) at (-9, 6) {};
			\node[V, label=right:\large $v_4$] (v4) at (-9, 2) {};
			\node at (-11, -2) {\LARGE $G_i$};
						\draw (v1) -- (v2);
			\node at (-13.35, 4) {\Large $e_1$};
			\draw (v2) -- (v3);
			\node at (-11, 6.35) {\Large $e_2$};

						\node[V, label=left:\large $z$, fill=red] (z) at (3, 0) {};
			\node[V, label=left:\large $v_1$, fill=red] (v12) at (1, 2) {};
			\node[V, label=left:\large $w_{e_1}$] (we1) at (0, 4) {};
			\node[V, label=left:\large $v_2$] (v22) at (1, 6) {};
			\node[V, label=above:\large $w_{e_2}$] (we2) at (3, 7) {};
			\node[V, label=right:\large $v_3$, fill=red] (v32) at (5, 6) {};
			\node[V, label=right:\large $w_{e_3}$] (we3) at (6, 4) {};
			\node[V, label=right:\large $v_4$] (v42) at (5, 2) {};
			\node at (3, -2) {\LARGE $G_i'$};
						\draw (v12) -- (v22);
			\draw (v12) -- (we1);
			\draw (v22) -- (we1);
			\draw (v22) -- (v32);
			\draw (v22) -- (we2);
			\draw (v32) -- (we2);
			\draw (z) -- (v42);
			\draw (z) -- (we3);

						\draw[->, line width=3pt] (-5.5,3) -- (-3.5,3);

		\end{tikzpicture}
		\caption{A sketch of the reduction from \prb{Nonuniform \CT} to \prb{Nonuniform \DT} in~\Cref{LemmaVCtoDStimeline}.
For each snapshot $G_i$ of~$\mathcal{G}$ one snapshot $G_i'$ is introduced for $\mathcal{G}'$.
In the example shown in the figure, it is assumed that the underlying graph of $\mathcal{G}$ has the edge set $E_\downarrow = \{e_1, e_2, e_3\}$.
For each of these edges one vertex is introduced (with zero activity intervals) and connected to the respective endpoints in $G_i'$ if and only if the edge appears in $G_i$.
Further, all the isolated vertices are connected to $z$, which is assigned enough activity intervals to be active in every snapshot.
The vertices marked in red are active.}
		\label{FigureLemmaVCtoDStimeline}
	\end{figure}
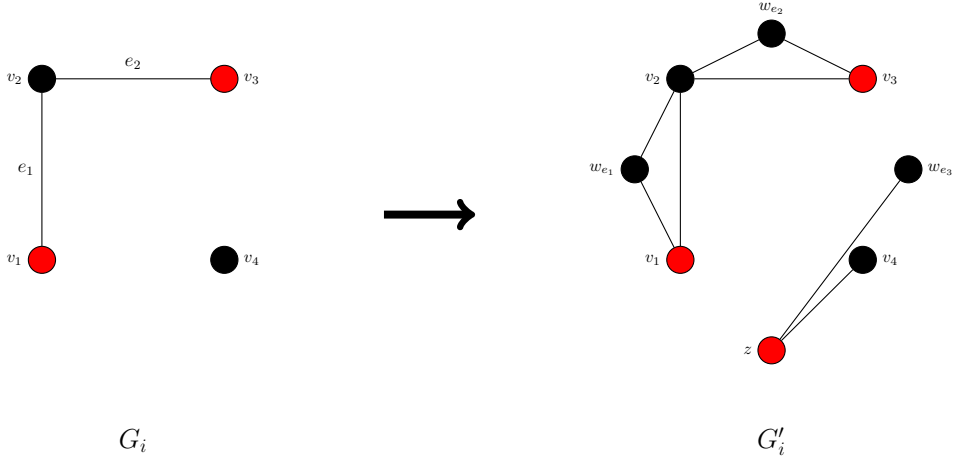

	\textit{Correctness:} We show that $(\mathcal{G}, (k_v)_{v \in V}, \ell = 1)$ is a yes-instance of \prb{Nonuniform \CT} if and only if $(\mathcal{G}', (k'_v)_{v \in V'}, \ell = 1)$ is a yes-instance of \prb{Nonuniform \DT}.

	($\Rightarrow$) Suppose the activity timeline $\mathcal{T}$ is a solution to $(\mathcal{G}, (k_v)_{v \in V}, \ell)$.
In particular,~$\mathcal{T}$ covers all temporal edges in $\mathcal{G}$.
Consider
	\begin{align*}
		\mathcal{T}' :=
		\begin{cases}
			\mathcal{T} \cup \bigcup_{i = 1}^{\frac{T}{2}} \{(z, 2i -1 , 2i)\} & \text{ if } T \text{ is even,} \\
			\mathcal{T} \cup \bigcup_{i = 1}^{ \left\lfloor \frac{T}{2} \right\rfloor } \{(z, 2i - 1, 2i)\} \cup \{(z, T, T)\} & \text{ if } T \text{ is odd}.
\\
		\end{cases}
	\end{align*}
	Intuitively, $\mathcal{T}'$ is basically $\mathcal{T}$ and additionally the vertex $z$ is active in each snapshot.
	The constraints $(k'_v)_{v \in V'}$ for the number of activity intervals of each vertex are clearly satisfied.
Since $z$ is active in every snapshot, it remains to show that all temporal vertices, which are not a neighbor of $z$, are dominated.
Assume $(v,i)$ is such a temporal vertex.
We have to distinguish between two cases.
If $v \in V$ and~$v \notin N_{G_i} (z)$, then $uv \in E_i$ for some $u \in V$.
This implies that $u$ or $v$ is active in snapshot $G_i$ with respect to $\mathcal{T}$ (because the temporal edge $(uv,i)$ has to be covered) and therefore~$(v,i)$ is dominated in $\mathcal{G}'$ because $\mathcal{T} \subseteq \mathcal{T}'$.
If $v = w_e$ for some $e \in E_i$ and $v \notin N_{G_i} (z)$, then by construction $w_eu, w_eu' \in E_i'$ for some $u,u' \in V$ where~$e = uu' \in E_i$.
In particular, $u$ or $u'$ is active with respect to $\mathcal{T}$ and therefore $(v,i)$ is dominated in $\mathcal{G}'$ by $\mathcal{T}'$.

	($\Leftarrow$) Suppose the activity timeline $\mathcal{T}'$ is a solution to $(\mathcal{G}', (k'_v)_{v \in V'}, \ell)$ and define
	\begin{align*}
		\mathcal{T} := \{(v, a, b) \in \mathcal{T}' : v \in V\},
	\end{align*}
	that is, $\mathcal{T}$ is simply the set of activity intervals of vertices from $V$, which are contained in $\mathcal{T}'$.
By construction, it is clear that the constraints on the number of activity intervals are satisfied.
We show that $\mathcal{T}$ covers all temporal edges in $\mathcal{G}$.
If~$e=uv \in E_i$, then $w_e$ is dominated by $u$ or $v$ in $G_i'$, because $w_e$ is not adjacent to $z$ in $G_i'$ and has no own activity intervals.
So in particular, $u$ or $v$ is active in $G_i'$ with respect to $\mathcal{T}'$.
By definition $\mathcal{T}$ contains the same activity intervals of $u$ and~$v$ as $\mathcal{T}'$.
It follows that $u$ or $v$ is active in $G_i$ with respect to $\mathcal{T}$ and therefore the temporal edge~$(e,i)$ is covered by $\mathcal{T}$.
\end{proof}

With~\Cref{LemmaVCtoDStimeline} we are now ready to prove~\Cref{W1hardnessN}, that is, the $\W[1]$-hardness of \DT with respect to $n$, even if $\ell = 1$.

\begin{proof}[Proof of \Cref{W1hardnessN}]
	We want to make use of~\Cref{LemmaVCtoDStimeline}, this means we provide a reduction from \prb{Nonuniform \DT} with~$\ell=1$ to \DT with~$\ell'=1$.
A sketch of the reduction can be found in~\Cref{FigureW1hardnessN}.

	Let $I:= (\mathcal{G} = (G_1, \dots, G_T), (k_v)_{v \in V}, \ell = 1)$ be an instance of \prb{Nonuniform \DT} with vertex set $V := V(\mathcal{G})$.
The maximum number of permitted activity intervals $k' := \max_{v \in V} k_v$ will be the permitted number of activity intervals in the \DT instance~$I':=(\mg',k',\ell' = 1)$.
We construct a new temporal graph $\mathcal{G}'$ on the vertex set $$V' := V(\mathcal{G}') = V  \cup \{v^* : v \in V\}.$$ For this we introduce $T + 2k'$ snapshots to $\mathcal{G}'$ where snapshot $G_i'$  has the edge set
	\begin{align*}
		E_i' :=
		\begin{cases}
			\bigcup_{v \in V} \{vv^*\} \cup \bigcup_{vw \in E_i} \{vw, vw^*, v^*w, v^*w^*\} & \text{ if } i \in [T], \\
			\{vv^* : v \in V \text{ such that }  i > T + 2(k'-k_v)\} & \text{ if } i \in [T+1, T+2k'].
		\end{cases}
	\end{align*}
	Intuitively, for $i \in [T]$ the snapshot $G_i'$ is constructed from $G_i$ by taking $G_i$ and a copy of $G_i$ and making $v$ and its copy $v^*$ adjacent to the same vertices.
In the next $2(k'-k_v)$ snapshots $(G'_{T+1}, \dots, G'_{T + 2(k'-k_v)})$ the vertex $v$ and its copy $v^*$ remain isolated and in the last $2k_v$ snapshots $(G'_{T + 2(k'-k_v) + 1}, \dots, G'_{T + 2k'})$ they are only incident with the edge $vv^*$.
Consequently,~$v$ and $v^*$ dominate the same set of vertices in each every snapshot of $\mathcal{G}'$.
Further, observe that for $i \in [T]$ a set of vertices from $V$ is a dominating set of $G_i$ if and only if it is a dominating set of~$G_i'$.
	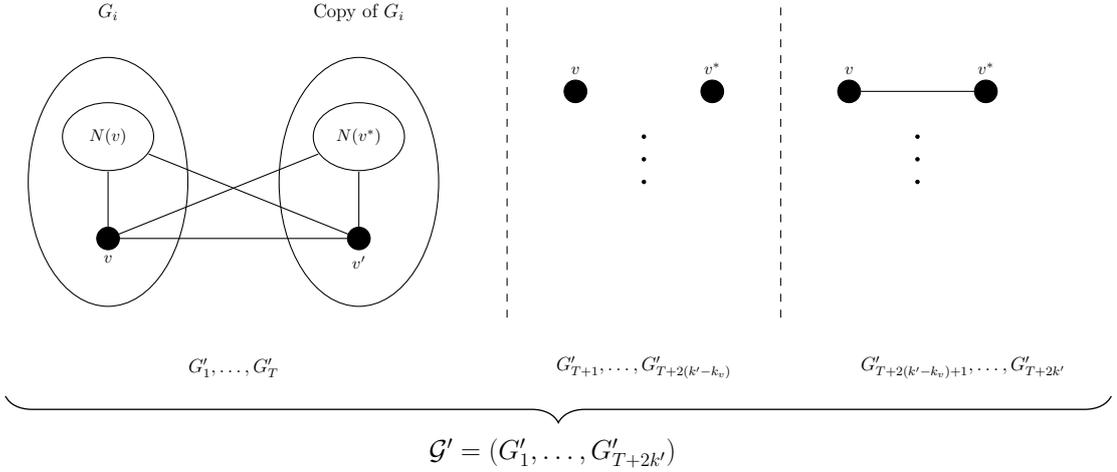
\begin{figure}[t]
		\centering
		\begin{tikzpicture}[scale = 0.6, transform shape,
			V/.style = {circle, draw, fill=black, inner sep = 5pt},
			d/.style = {circle, draw, fill=black, inner sep = 0.75pt},
			d2/.style = {circle, draw, fill=black, inner sep = 1.5pt}
			]

						\draw (-2.75,3) ellipse (1.75cm and 2.75cm);
			\node at (-2.75, 6.75) {\Large $G_i$};
			\draw (2.75,3) ellipse (1.75cm and 2.75cm);
			\node at (2.75, 6.75) {\Large Copy of $G_i$};
			\draw (-2.75,4) ellipse (1cm and 0.75cm);
			\node at (-2.75,4) {\large $N(v)$};
			\draw (2.75,4) ellipse (1cm and 0.75cm);
			\node at (2.75,4) {\large $N(v^*)$};
			\node[V, label=below:\large $v$] (v) at (-2.75,1.75) {};
			\node[V, label = below:\large $v'$] (v') at (2.75,1.75) {};
						\draw (v) -- (v');
			\node[inner sep= 0] (d1) at ($(-2.75,4)+(-30:1 and 0.75)$) {};
			\node[inner sep=0] (d2) at ($(2.75,4)+(210:1 and 0.75)$) {};
			\node[inner sep= 0] (d3) at ($(-2.75,4)+(-90:1 and 0.75)$) {};
			\node[inner sep=0] (d4) at ($(2.75,4)+(-90:1 and 0.75)$) {};
			\draw (v) -- (d2);
			\draw (v') -- (d1);
			\draw (v) -- (d3);
			\draw (v') -- (d4);
			\node at (0, -1.1) {\Large $G'_1, \dots, G'_{T}$};

						\node[V, label=above:\large $v$] (v2) at (7.5,5) {};
			\node[V, label = above:\large $v^*$] (v'2) at (10.5,5) {};
			\node[d] () at (9,4) {};
			\node[d] () at (9,3.5) {};
			\node[d] () at (9,3) {};
			\node at (9, -1.1) {\Large $G'_{T + 1}, \dots, G'_{T + 2(k'-k_v)}$};

						\node[V, label=above:\large $v$] (v3) at (13.5,5) {};
			\node[V, label = above:\large $v^*$] (v'3) at (16.5,5) {};
			\node[d] () at (15,4) {};
			\node[d] () at (15,3.5) {};
			\node[d] () at (15,3) {};
						\draw (v3) -- (v'3);
			\node at (16, -1.1) {\Large $G'_{T + 2(k'-k_v) + 1}, \dots, G'_{T + 2k'}$};

						\draw [decorate,decoration={brace,amplitude=10pt,mirror,raise=4pt}, line width=0.5pt] (-5,-1.5) -- (19.25,-1.5);
			\node at (7, -3) {\LARGE $\mathcal{G}' = (G'_1, \dots, G'_{T + 2k'})$};

						\draw[dashed] (6, 0) -- (6, 7);
			\draw[dashed] (12, -0) -- (12, 7);

		\end{tikzpicture}
		\caption{A sketch of the reduction from \prb{Nonuniform \DT} to \DT in~\Cref{W1hardnessN}.
If $\mathcal{G} = (G_1, \dots, G_T)$ is the given temporal, then the newly constructed temporal graph $\mathcal{G}'$ consists of $T + 2k'$ snapshots, where $k'$ is the maximum number of permitted activity intervals in the nonuniform variant.
The reduction introduces for each vertex $v$ a copy $v^*$ and the figure sketches how these two vertices are linked in each snapshot.
Note that the number $k_v$ of permitted activity intervals may be different for each vertex $v \in V(\mathcal{G})$.}
		\label{FigureW1hardnessN}
	\end{figure}

	\textit{Correctness:} We show that $I$ is a yes-instance of \prb{Nonuniform \DT} if and only if $I'$ is a yes-instance of \DT.
	Recall that~$\ell = \ell' = 1$.

	($\Rightarrow$) Suppose the activity timeline $\mathcal{T}$ is a solution for $I$ and consider
	\begin{align*}
		\mathcal{T'} := \mathcal{T} \cup \bigcup_{v \in V} \Big(
		& \bigcup_{i = 1}^{k'-k_v} \{(v,T + 2i-1,T + 2i), (v', T + 2i-1, T+2i)\} \\
		& \cup \bigcup_{i = k'-k_v + 1}^{k'} \{(v', T + 2i-1, T+2i)\} \Big).
	\end{align*}
	The activity timeline $\mathcal{T}'$ has the same activity intervals as $\mathcal{T}$ in the first $T$ snapshots.
Moreover, for each vertex $v \in V$ the vertices $v$ and $v^*$ are both active in the $2(k'-k_v)$ snapshots $G'_{T+1}, \dots, G'_{T + 2(k'-k_v)}$ and $v^*$ is also active in the last $2k_v$ snapshots~$G'_{T + 2(k'-k_v) + 1}, \dots, G'_{T + 2k'}$.
	Note that each vertex has at most $k'$ activity intervals (since $v$ has at most $k_v$ activity intervals in $\mathcal{T}$) and each of them is of length at most one.
It remains to show that $\mathcal{T}'$ dominates all temporal vertices in $\mathcal{G}'$.
For $i \in [T]$, any dominating set in $G_i$ is also a dominating set in $G_i'$ and therefore~$\mathcal{T}'$ dominates all temporal vertices in the first $T$ snapshots (because $\mathcal{T}$ does).
For $i \in [T+1, T + 2(k'-k_v)]$, the vertices $v$ and $v^*$ are both active and therefore dominated.
For $i \in [T + 2(k'-k_v)+1, T+ 2k']$ they are dominated by $v^*$, because the edge $vv^*$ exists and $v^*$ is active.
Consequently, all temporal vertices are dominated and $\mathcal{T}'$ is a valid solution to $I'$.

	($\Leftarrow$) Now suppose $\mathcal{T}'$ is a solution to $I'$.
Informally, the activity timeline $\mathcal{T}$ for the nonuniform instance will contain the activity intervals of a vertex~$v$, which correspond to the parts of activity intervals of $v$ or $v^*$, which intersect the first $T$ snapshots.
Formally, we set
	\begin{align*}
		\mathcal{T} := \{(v,a,\min (b,T)) : (v,a,b) \in \mathcal{T}' \text{ or } (v^*,a,b) \in \mathcal{T}' \text{ where } v \in V, a \leq T\}.
	\end{align*}
	We show that in total at most $k_v$ activity intervals of $v$ and $v^*$ intersect the first $T$ snapshots.
Recall that the activity intervals have length one and therefore each of them intersects exactly two snapshots.
In the last $2k_v-1$ snapshots the vertices $v$ and $v^*$ are only adjacent to each other and therefore in total $k_v$ activity intervals from them have to intersect these snapshots.
Since they have length one, none of them can intersect the $2(k'-k_v) - 1$ snapshots $G'_{T+2}, \dots , G'_{T + 2(k'-k_v)}$.
In these snapshots both $v$ and $v^*$ are isolated, so $(k'-k_v)$ activity intervals of each of them have to intersect these snapshots.
Again they have length one, which implies that they cannot intersect the first $T$ snapshots.
Consequently, there are in total $2(k'-k_v) + k_v$ activity intervals of $v$ and $v^*$, which do not intersect the first $T$ snapshots.
It follows that $\mathcal{T}$ contains at most $k_v$ activity intervals of $v$.
We show that $\mathcal{T}$ dominates all temporal vertices in $\mathcal{G}$.
For this recall that $v$ and $v^*$ dominate the exact same vertices in every snapshot $G_1', \dots, G_T'$ and that a set of vertices from $V$ is a dominating set of $G_i$ if and only if it is a dominating set of~$G_i'$.
\end{proof}

The $\W[1]$-hardness results for both problems do not apply to the case~$\ell=0$. For \CT, this case is fixed-parameter tractable because there exists an ILP formulation in which the number of variables is bounded by some function of~$n$~\cite[Lemma~10]{DBLP:journals/mst/FroeseKZ24}. It is also straightforward to extend this ILP formulation to \PCT.
To complete the picture, we extend the ILP formulation also to \PDT, thus showing that it  is fixed-parameter tractable for~$\ell=0$ as well.

\begin{theorem}

	\PDT is fixed-parameter tractable with respect to $n$ if $\ell = 0$.
\end{theorem}

\begin{proof}
	Suppose $(\mathcal{G} = (G_1, \dots, G_T), k, \ell)$ is an instance of \PDT with~$\ell = 0$.
We provide an ILP formulation (adaption from \cite{DBLP:journals/mst/FroeseKZ24}), where the number of variables is upper-bounded by some function of $n$.
For this we define the following
	\begin{itemize}
		\item variables $X_E^S$ with $S \subseteq V$ and $E \subseteq \binom{V}{2}$, which correspond to the number of times the vertex set $S$ is active in a snapshot with edge set $E$,
		\item constants $\alpha (E)$ storing the number of snapshots with edge set $E$,
		\item constants $t_E^S$ storing the number of dominated vertices, if $S$ is the set of active vertices and the snapshot has edge set $E$.
	\end{itemize}
	Note that the number of these variables and constants only depends on $n$.
In particular, there are $2^{\binom{n}{2} + n}$ variables $X_E^S$.
With this we are now able to formulate our integer linear program:
	\begin{align}
		\sum_{E \subseteq \binom{V}{2}} \sum_{S \subseteq V} t_E^S X_E^S &\geq t \tag{1}\\
		\sum_{S \subseteq V} X_E^S &= \alpha (E), && \text{ for all } E \subseteq \binom{V}{2}  \tag{2}\\
		\sum_{E \subseteq \binom{V}{2}} \sum_{\substack{S \subseteq V : \\ v \in S}} X_E^S &\leq k, && \text{ for all } v \in V  \tag{3}\\
		X_E^S &\in \mathbb{N}_0  \tag{4}
	\end{align}
	\textit{Correctness:} We show that this ILP is feasible if and only if $(\mathcal{G}, k, \ell = 0)$ is a yes-instance of \PDT.

	($\Rightarrow$) Suppose a solution to the ILP is given.
We construct an activity timeline~$\mathcal{T}$ for $(\mathcal{G}, k, \ell = 0)$.
For this let $S_1, S_2, \dots$ be some fixed order on the subsets of $V$.
Now use $X_E^{S_1}$ activity intervals of exactly the vertices from $S_1$ in the first $X_E^{S_1}$ snapshots with edge set $E$.
Then use $X_E^{S_2}$ activity intervals of exactly the vertices from $S_2$ for the next $X_E^{S_2}$ snapshots with edge set $E$, and so on.
By doing this for every edge set~$E$, we obtain a $k$-activity $0$-timeline $\mathcal{T}$.
Constraint (1) of the ILP ensures that at least $t$ temporal vertices are dominated, while constraint (3) ensures that at most $k$ activity intervals of each vertex are selected.
Finally (2) guarantees that the number of snapshots with edge set $E$ is respected.

	($\Leftarrow$) Consider an activity timeline $\mathcal{T}$, which is a solution to $(\mathcal{G}, k, \ell = 0)$.
Set~$X_E^S$ to the number of times the vertex set $S$ is active in a snapshot with edge set $E$ (with respect to the solution $\mathcal{T}$).
Condition (1) of the ILP is satisfied, because $\mathcal{T}$ dominates at least $t$ temporal vertices, and condition (3) holds, because $\mathcal{T}$ contains at most $k$ activity intervals of each vertex.
Finally, the conditions (2) and (4) are direct consequences of the definition of~$X_E^S$.
\end{proof}

Finally, we show that both \PDT and \PCT can be solved in $2^{\Oh(t)}\cdot (n+T)^{\Oh(1)}$~time, where~$t$ denotes the number of temporal vertices/edges, which need to be dominated/covered.
Recall that Dondi et al.~\cite[Theorem~6]{DondiPartialBounded} provided an FPT-algorithm for \PCT with running time~$t^t \cdot (n+T)^{\Oh(1)}$ if $k =1$.
Hence, our approach improves upon their running time and is not limited to~$k=1$.
The idea is to use color coding~\cite{AYZ95}, a very popular technique to obtain FPT-algorithms~\cite{CFK+15}.
To provide deterministic algorithm, we use \emph{$(n, k)$-perfect hash families} which can be computed efficiently~\cite{DBLP:conf/focs/NaorSS95}.

Initially, we show our result for \PDT.

\begin{theorem}
	\label{DomTimelineParametert}

	\PDT is solvable in $2^{\Oh(t)} \cdot (n+T)^{\Oh(1)}$~time.\end{theorem}
\begin{proof}
	Let $(\mathcal{G} = (G_1, \dots, G_T), k, \ell, t)$ be a \PDT-instance.
We assume that $T \geq k(\ell+ 1)$ (otherwise the instance is trivial) and fix some order~$v_1, \dots , v_n$ of the vertices from $V(\mathcal{G})$.
The idea is to use color coding, so we assume that our temporal vertices are colored randomly with $t$ colors.
We say that an activity timeline~$\mathcal{T}$ dominates a color, if it dominates a temporal vertex of the respective color and we denote with $C_{a,b}^j$ the colors, which are dominated by the activity interval~$(v_j, a, b)$.
Now, consider the following dynamic programming table.

\emph{Table Definition:}
For a set~$S \subseteq \{1, \dots, t\}$ of colors, and integers $j \in [0, n], k' \in [0, k]$ we define
	\begin{align*}
		\DP [S, j, k'] \,\, \hat{=} \,\, \texttt{true}
							\end{align*}
	if and only if there exists a $k$-activity $\ell$-timeline $\mathcal{T} \subseteq \{v_1, \dots, v_j\} \times [T] \times [T]$ such that
	\begin{enumerate}[label=(\roman*)]
		\item $|\mathcal{T} \cap (\{v_j\} \times [T] \times [T]) | = k'$, and
		\item $\mathcal{T}$ dominates the colors from $S$.
	\end{enumerate}
	In other words, $\DP [S, j, k'] \,\, \hat{=} \,\, \texttt{true}$ if and only if the colors from $S$ can be dominated by using $k$ activity intervals of $v_1, \dots, v_{j-1}$ and $k'$ activity intervals of $v_j$.

	\emph{Initialization:}
	We set the entry $\DP [S, 0, k]$ to \texttt{true} if and only if~$S = \emptyset$.

	\emph{Recurrence:}
We consider the case where~$k'=0$ and~$k'> 0$ individually:
	\begin{align*}
		\DP [S,j,k'] =
		\begin{cases}
			\bigvee_{a \in [T - \ell]}	\DP [S \setminus C_{a,a + \ell}^j, j, k' - 1], &\, \text{ if } k' > 0, \\
			\DP [S, j-1, k],  &\, \text{ if } k' = 0.
		\end{cases}
	\end{align*}

	\emph{Correctness:}
	We prove the correctness of the computation by showing two directions.

	($\Rightarrow$) Suppose $\DP [S,j,k'] = \texttt{true}$.
By definition, there exists a $k$-activity $\ell$-timeline $\mathcal{T} \subseteq \{v_1, \dots, v_j\} \times [T] \times [T]$, which satisfies (i) and (ii).
Otherwise, if $k' > 0$, then there exists an activity interval $(v_j, a, a + \ell) \in \mathcal{T}$.
Consequently, $\mathcal{T} \setminus \{(v_j, a , a + \ell)\}$ is considered in the definition of $\DP [S \setminus C_{a,a + \ell}^j, j, k' - 1]$ and satisfies (i) and (ii).
Therefore, we have $\DP [S \setminus C_{a,a + \ell}^j, j, k' - 1] = \texttt{true}$.
If $k' = 0$, then $\mathcal{T}$ only contains activity intervals of $v_1, \dots, v_j$.
This means that $\mathcal{T}$ considered in the definition of the table entry $\DP [S, j-1, k]$ and subsequently $\DP [S, j-1, k] = \texttt{true}$.

	($\Leftarrow$) Let $k ' > 0$ and $\DP [S \setminus C_{a,a + \ell}^j, j, k' - 1] = \texttt{true}$ for some $a \in [T - \ell]$.
If $\mathcal{T}$ is considered in the definition of $\DP [S \setminus C_{a,a + \ell}^j, j, k' - 1]$ and satisfies (i) and (ii), then the activity timeline $\mathcal{T} \cup \{(v_j,a,a+ \ell)\}$ is considered in the definition of $\DP [S,j,k']$ and also satisfies (i) and (ii).
This implies $\DP [S,j,k'] = \texttt{true}$.
If $k' = 0$, then each $k$-activity $\ell$-timeline, which is considered in the definition of $\DP [S, j-1, k]$ and satisfies (i) and (ii), is also considered in the table entry $\DP [S, j, k'=0]$ and satisfies (i) and (ii) for this entry.
Therefore, $\DP [S, j-1, k] = \texttt{true}$ implies $\DP [S, j, k] = \texttt{true}$.

	\emph{Result:}
	Finally, we return yes if and only if $\DP [\{1, \dots, t\}, n, k] = \texttt{true}$.
In particular, we return yes if and only if there exists a $k$-activity $\ell$-timeline, which dominates all colors from $S=\{1, \ldots, t\}$.

	\emph{Running Time:}
	The table has size $\Oh(2^t \cdot n \cdot k)$ and each update takes $\Oh(n \cdot T^2)$ time (to determine the set $C_{a,a+ \ell}^j$ for all $a \in [T]$).
This yields a running time of $\Oh(2^t \cdot n^2 \cdot T^2 \cdot k)$.
In the beginning of the algorithm we assumed that some random coloring of the temporal vertices is given.
Instead of using random colorings, we can use derandomization based on $(n, k)$-perfect hash families; see~\cite{CFK+15} for details.
This guarantees that for each set of $t$ temporal vertices, we deterministically find a coloring, which colors these $t$ temporal vertices with exactly $t$ colors.
This means, there is a solution to $(\mathcal{G}, k, \ell, t)$ if and only if the algorithm from above returns yes for one of the constructed colorings.
So, running the algorithm for each of the colorings of the construction, solves our problem instance. 
The overall running time can then be bounded by $(2e)^t t^{\Oh(\log t)} \cdot n^3\log (n)\cdot T^3k$, which is $2^{\Oh(t)} \cdot (n+T)^{\Oh(1)}$~time..
\end{proof}

The algorithm of \Cref{DomTimelineParametert} also works analogously for \PCT: We only need
to color the edges instead of the vertices.

\begin{corollary}
	\label{CovTimelineParametert}
	\PCT is solvable in $2^{\Oh(t)} \cdot (n+T)^{\Oh(1)}$~time.
\end{corollary}

\section{Conclusion}
\label{chapter:conclusion}

We studied the classical and parameterized complexity of \PCTT and \PDTT.
We showed that all problems admit FPT-algorithms for~$\vimw +\,k + \ell$, where~$\vimw$ is  the vertex-interval-membership width.
Our running time bounds also give XP-algorithms for~$\vimw$ alone.
For \CT this improves upon an XP-algorithm for~$n$~\cite{DBLP:journals/mst/FroeseKZ24}.
Moreover, we showed that the smaller parameter~$\imw+\,k+\ell$ leads to para-NP-hardness for all problems except \DT.
Hence, we discovered a parameter where a dominating set problem is tractable while the corresponding vertex cover variant is not.

For future work it is interesting to investigate whether already the parameter~$\vimw+\,k$ yields an FPT-algorithm for any problem in our study.
Moreover, it is open whether \PCT and \PDT admit an FPT-algorithm with respect to~$n + k$.
Even for~$k=1$ fixed-parameter tractability with respect to~$n$ is still open.
Moreover, one could investigate the complexity with respect to other temporal graph parameters such as the temporal neighborhood diversity~\cite{EHLM24} or temporal graph parameters that are sensitive to the order of the snapshots~\cite{imwVariants,CMW24}.

A further problem related to \CT is \prb{MinTimeline$_+$}.
In this problem, the value of~$\ell$ bounds the sum of the lengths of the activity intervals instead of the length of the longest activity interval.
\prb{MinTimeline$_+$} was also introduced by Rozhenstein et al.~\cite{DBLP:journals/datamine/RozenshteinTG21} and studied in several works~\cite{DBLP:journals/mst/FroeseKZ24, DondiBoundedDegree, DL25, DondiInsightsDisentangling, DondiPopaExactApprox, SchubertMA,LMZD24}.
To distinguish this problem from the other problems in our naming convention, \CT (\prb{MinTimeline$_\infty$}) could be called \prb{MinMax \CT} and
\prb{MinTimeline$_+$} could be called \prb{Sum \CT}, and the timeline variants of \DS could be named analogously.
It is open which of our positive results for~$\vimw$ and~$\imw$ can be transferred to the \prb{Sum} variants of the problems.
Finally, it is interesting to study other classic problems as \textsc{Feedback Vertex Set} in the \prb{MinMax Timeline} and \prb{Sum Timeline} setting.

\bibliographystyle{plain}
\bibliography{bibliography}

\end{document}